\documentclass[runningheads,envcountsame]{llncs}
\usepackage[T1]{fontenc}

\usepackage{amsthm}
\usepackage{amsmath, amssymb}
\usepackage{thmtools,thm-restate}

\usepackage{tikz}
\usepackage{xspace}

\usepackage{booktabs}   
\usepackage{subcaption} 

\usepackage{enumitem} 
\usepackage{makecell} 
\usepackage{mathtools}
\usepackage{wrapfig}

\usepackage{amsmath}
\usepackage{pifont}
\usepackage{thm-restate}
\usetikzlibrary{arrows,automata,shapes,decorations,decorations.markings,calc, matrix,decorations.pathmorphing, patterns,backgrounds,shapes.misc,arrows.meta,positioning}

\usepackage{xspace}
\usepackage[ruled,vlined,resetcount,linesnumbered,noend]{algorithm2e}
\usepackage{parskip}
\usepackage{comment}

\usepackage{arydshln}
 \usepackage{multicol}
 \usepackage{multirow}

\usepackage{hyperref}
\usepackage[capitalise]{cleveref}

\usepackage{varwidth}

\usepackage{array}
\newcolumntype{H}{>{\setbox0=\hbox\bgroup}c<{\egroup}@{}}

\usepackage{ifthen}
\usepackage{stackengine}

\usepackage[normalem]{ulem} 

\usepackage{todonotes}
  \newcounter{todocounter}

\newcommand\xvar{x}

\newcommand\varset{{\mathbb X}}

\newcommand\valset{{\mathbb V}}

\newcommand\lbl\ell
\newcommand\action\lbl

\newcommand\run\rho

\colorlet{colorPO}{darkgray!80!black}
\colorlet{colorRF}{blue}
\colorlet{colorOB}{orange}
\colorlet{colorCO}{red!80!black}
\colorlet{colorMO}{red!80!black}
\colorlet{colorFR}{olive}
\colorlet{colorECO}{orange}
\colorlet{colorCOM}{magenta}
\colorlet{colorSW}{teal}
\colorlet{colorHB}{green!40!black}
\colorlet{colorPPO}{magenta}
\colorlet{colorRSEQ}{green!40!black}
\colorlet{colorSC}{violet}
\colorlet{colorHBMO}{brown}
\colorlet{colorPSC}{violet}
\colorlet{colorREL}{olive}
\colorlet{colorWB}{olive}
\colorlet{colorRMW}{brown}

\newcommand\conf\gamma

\tikzset{
   every path/.style={>=stealth},
   po/.style={->,draw=colorPO,shorten >=-0.5mm,shorten <=-0.5mm},
   ppo/.style={->,draw=colorPPO,shorten >=-0.5mm,shorten <=-0.5mm},
   sw/.style={->,draw=colorSW,shorten >=-0.5mm,shorten <=-0.5mm},
   rf/.style={->,draw=colorRF,dashed,shorten >=-0.5mm,shorten <=-0.5mm},
   mo/.style={->, draw=colorMO,dotted,shorten >=-0.5mm,shorten <=-0.5mm},
    mrf/.style={->,draw=colorRF,dashed,shorten >=-0.5mm,shorten <=-0.5mm},   
   urf/.style={->,draw=colorRF,shorten >=-0.5mm,shorten <=-0.5mm},
   fr/.style={->,draw=colorFR,dashed,shorten >=-0.5mm,shorten <=-0.5mm},
   hb/.style={->,draw=colorHB,thick,shorten >=-0.5mm,shorten <=-0.5mm},
   co/.style={->,draw=colorCO,dotted,very thick,shorten >=-0.5mm,shorten <=-0.5mm},
   rmw/.style={->,draw=colorRMW,thick,shorten >=-0.5mm,shorten <=-0.5mm},
   rseq/.style={->,draw=colorRSEQ,thick,dotted,shorten >=-0.5mm,shorten <=-0.5mm},
   com/.style={->,draw=colorCOM,thick,shorten >=-0.5mm,shorten <=-0.5mm},
}

\newcommand{\tup}[1]{\langle#1\rangle}

\newcommand{\R}{\mathsf{R}}
\newcommand{\W}{\mathsf{W}}
\newcommand{\E}{\mathsf{E}}

\newcommand{\wramm}{\mathsf{WRA}}
\newcommand{\ramm}{\mathsf{RA}}
\newcommand{\rlxmm}{\mathsf{Relaxed}}
\newcommand{\sramm}{\mathsf{SRA}}

\newcommand{\cmmm}{\mathsf{CM}}
\newcommand{\ccmm}{\mathsf{CC}}
\newcommand{\cvmm}{\mathsf{CCv}}

\newcommand{\mmorder}{\preccurlyeq}

\newcommand{\tuple}[1]{\left\langle#1\right\rangle}

\newcommand{\po}{\mathsf{\color{colorPO}po}}

\newcommand{\rf}{\mathsf{\color{colorRF}rf}}
\newcommand{\rfinv}{\rf^{-1}}

\newcommand{\mo}{\mathsf{\color{colorMO}mo}}

\newcommand{\hb}{\mathsf{\color{colorHB}hb}}

\newcommand\locx[1]{{{#1}_{x}}}
\newcommand{\ob}[1]{\mathsf{{\color{colorOB}ob}_{#1}}}

\newcommand{\irr}{\mathsf{irr}}
\newcommand{\acy}{\mathsf{acy}}

\newcommand{\op}{\mathsf{op}}

\newcommand{\event}{e}

\newcommand{\lloc}{\mathsf{var}}

\newcommand{\val}{\mathsf{val}}

\newcommand{\rd}{\mathtt{r}}
\newcommand{\wt}{\mathtt{w}}

\newcommand{\ex}{\mathsf{\color{black}X}}

\newcommand{\expartial}{\overline{\ex}}

\newcommand{\ypartial}{\overline{\mathsf{\color{black}Y}}}

\newcommand{\NP}{\mathsf{NP}}

\newcommand\np{\mathsf{NP}}

\newcommand{\MemModel}{\mathcal{M}}

\newcommand{\var}{\operatorname{var}}

\newcommand{\xra}{\xrightarrow}

\SetKwProg{myfun}{function}{}{}
\SetKwProg{myhandler}{handler}{}{}
\SetKwFunction{init}{Initialization}
\SetKwFunction{getLW}{getLastWriteBeforeOffline}
\SetKwFunction{poprop}{poPropagate}
\SetKwFunction{joinchangedindex}{pairwiseMaxIndexes}
\SetKwFunction{getLWB}{getLastWriteBefore}
\SetKwFunction{getLRB}{getLastReadBefore}
\SetKwFunction{getLRWB}{getLastReadWriteBefore}
\SetKwFunction{initWtLst}{initWtLst}
\SetKwFunction{initRdLst}{initRdLst}
\SetKwFunction{initWtRdLst}{initWtRdLst}
\SetKwFunction{get}{get}
\SetKwFunction{getHBTS}{getHBTimestamps}
\SetKwFunction{getTCPC}{getTopAndPositionInRFChain}
\SetKwFunction{checkRace}{checkRace}
\SetKwFunction{rdhandler}{readHandler}
\SetKwFunction{wthandler}{writeHandler}
\SetKwFunction{acqhandler}{acquire}
\SetKwFunction{relhandler}{release}
\SetKwInput{Input}{Input}
\SetKwInOut{Output}{Output}
\SetKw{Let}{let}
\SetKw{Break}{break}
\SetKw{NOT}{not}
\SetKw{declare}{declare}
\SetKw{exit}{exit}
\SetKw{Continue}{continue}
\SetKw{Return}{return}
\SetKwFor{Foreach}{for each}{}{}%
\DontPrintSemicolon

\SetCommentSty{mycommfont}
\SetNoFillComment

\def\WCoh{write-coherence}
\def\SWCoh{strong-write-coherence}
\def\RCoh{read-coherence}
\def\WRCoh{weak-read-coherence}
\def\RlxWCoh{relaxed-write-coherence}
\def\RlxRCoh{relaxed-read-coherence}
\def\PORF{porf-acyclicity}
\def\OB{ob-acyclicity}

\tikzset{
   every path/.style={>=stealth},
   po/.style={->,color=colorPO, thick},
   ppo/.style={->,color=colorPPO, thick},
   sw/.style={->,color=colorSW, thick},
   rf/.style={->,color=colorRF,dashed, thick},
   mrf/.style={->,color=colorRF,dashed, thick},   
   urf/.style={->,color=colorRF, thick},
   fr/.style={->,color=colorFR,dashed, thick},
   hb/.style={->,color=colorHB,thick, thick},
   mo/.style={->,color=colorMO,dotted, very thick},
   dob/.style={->,color=colorDOB, very thick},
   ob/.style={->,color=colorOB, very thick},
   rmw/.style={->,color=colorRMW,thick, thick},
   rseq/.style={->,color=colorRSEQ,thick,dotted, thick},
   com/.style={->,color=colorCOM,thick, thick},
}

\makeatletter
\newcommand{\labitem}[2]{%
\def\@itemlabel{\textbf{#1}}
\item
\def\@currentlabel{#1}\label{#2}}
\makeatother

\SetKwFunction{ChCon}{CheckConsistency}
\SetKwFunction{ChPCon}{CheckPreemptionBoundedConsistency}
\SetKwComment{Comment}{/* }{ */}
\SetKwFunction{Explore}{Explore}
\SetKwFunction{ExploreSC}{ExploreSC}
\SetKwFunction{enable}{EnableEvent}
\SetKwFunction{enableth}{EnableThread}
\SetKwFunction{enablethsc}{EnableThreadSC}
\SetKwFunction{buildgraph}{BuildGraph}
\SetKwProg{Fn}{Function}{:}{}

\newcommand{\threewriter}{\textsc{3-Writer}}
\newcommand{\twowriter}{\textsc{2-Writer}}
\newcommand{\onewriter}{\textsc{1-Writer}}

\newcommand{\update}{\mathsf{Update}}
\newcommand{\Initialize}{\mathsf{Initialize}}
\newcommand{\CheckC}{\mathsf{CheckConsistency}}

\title{Complexity of Consistency Testing for the Release-Acquire Semantics}
\titlerunning{Complexity of Consistency Testing for the RA Semantics}
\author{
R.~Govind\inst{1,2,5}
\and
S.~Krishna\inst{3}
\and
Sanchari Sil\inst{4}
\and
B.~Srivathsan\inst{4,5}
}
\institute{
  The Institute of Mathematical Sciences, Chennai, India,
  \email{govind@imsc.res.in}\\
  \and
  Homi Bhabha National Institute, Anushaktinagar, Mumbai, India
  \and
  IIT Bombay, India, 
  \email{krishnas@ cse.iitb.ac.in}\\
  \and
   Chennai Mathematical Institute, India,
    \email{\{sanchari, sri\} @ cmi.ac.in}
  \and
  CNRS, ReLaX, IRL 2000, Chennai, India
}
\authorrunning{R.~Govind et al.}

\begin{document}

\maketitle

\begin{abstract}
    In a seminal work, Gibbons and Korach~\cite{Gibbons} studied the complexity of deciding whether an observed sequence of reads and writes of a multi-threaded program admits a sequentially consistent interleaving. They showed the problem to be $\NP$-hard even under strong syntactic restrictions. More recently, Chakraborty \emph{et al.}~\cite{ChakrabortyKMP24} considered the problem for weak memory models and proved that $\NP$-hardness remains even when the number of threads, the number of memory locations, and the value domain are all bounded.

    In this paper we revisit the problem for the release-acquire variants of the C11 memory model. Our main positive result is that consistency testing  can be done in polynomial-time when each memory location is written by at most one thread (multiple readers are allowed). Notably, this restriction is already $\NP$-hard for sequential consistency. We complement this upper bound with tight hardness results: the problem is $\NP$-hard when two threads may write to the same location, and allowing three writers per location rules out $2^{o(k)}\cdot n^{\mathcal{O}(1)}$ algorithms under the Exponential Time Hypothesis, where $k$ denotes the number of threads, and $n$ the number of memory operations.
\end{abstract}

\section{Introduction}
Multi-threaded programs, in which threads interact through shared memory operations, are prevalent across software systems. Modern hardware and compilers apply various optimizations (like instruction reordering, store buffering, etc.) to enhance performance. These optimizations can cause threads to observe writes to a shared memory location in an order that differs from the program order. \emph{Memory models} formalize the rules to constrain which values reads may observe, and allow the programmer to be aware of the memory behaviour. These models are implemented via hardware or software protocols. As the correctness of concurrent programs depends on the chosen memory model, a faithful implementation of memory models is essential. 

One of the fundamental results in this subject is the work of Gibbons and Korach~\cite{Gibbons} on \emph{testing shared memories}, which presents a way to test whether an observed execution of a multi-threaded program \emph{is consistent with} the (rules of the) underlying memory model. In their work~\cite{Gibbons}, Gibbons and Korach consider \emph{sequential consistency} (SC)~\cite{SC-Lamport78}, the most common and simple-to-understand memory model. SC uses an interleaving semantics to explain memory operations: any observed execution across the threads is identical to that of some global interleaving of the threads' operations that preserves the program order of each thread, and each read observes the value of the last write (in this interleaving) made to the memory location. Deciding whether an observed execution can be justified by such an interleaving is the classical \emph{consistency testing} problem for SC.\begin{center}
  \fbox{%
    \begin{minipage}{0.98\textwidth}
      \noindent \textsc{Verifying Sequential Consistency} (The VSC-problem~\cite{Gibbons})

      \medskip

      \noindent \textbf{Instance:} A variable set $\varset$, a value set $\valset$, a finite collection of nonempty sequences $S_1, S_2, \dots, S_k$, each consisting of a a finite set of memory operations of the form ``$read(x,d)$'' or ``$write(x,d)$'' with $x \in \varset$ and $d \in \valset$. 

      \medskip

      \noindent \textbf{Question:} Is there a sequence $S$, an interleaving of $S_1, S_2, \dots, S_k$ such that for each ``$read(x, d)$'' there is a preceding ``$write(x,d)$'' in $S$ with no other ``$write(x, d')$'' in between.
    \end{minipage}%
  }
\end{center}
The set of sequences $S_1, \dots, S_k$ forms our observed execution. A ``Yes'' answer to the question above indicates that the observed execution admits an SC behaviour, while a ``No'' answer shows the memory implementation can violate sequential consistency.

Gibbons and Korach~\cite{Gibbons} established that the VSC problem is $\NP$-complete, even under a number of natural syntactic restrictions (see Table~1 of~\cite{Gibbons}). Subsequent work has extended these complexity results along two complementary directions: (i) to weaker memory models that arise in programming languages, hardware, and distributed systems~\cite{ManovitH06,Qadeer03,ChenLHCSWP09,Mathur20,fine-grained-complexity,tso_rf_bui2021reads,BouajjaniEGH17}; and (ii) to algorithmic techniques for stateless model-checking of multi-threaded programs under sequential consistency~\cite{rv-equivalence}. Consistency testing also appears in the context of comparing memory models~\cite{Kokologiannakis:POPL:2017}. 

In this work, we revisit this consistency testing problem for the Release-Acquire semantics ($\ramm$), which form the core of the memory model adopted by the C11 programming language. We also investigate its variants, the strong release acquire ($\sramm$)~\cite{Lahav:2022} and the weak release acquire ($\wramm$)~\cite{KLSV:popl18}. For all these models, we make use of an axiomatic semantics that help reason about the observed executions as graphs and the rules of the model as constraints over special kinds of edges~\cite{Lahav2021,Lahav:2022}. Closely related to these memory models are $\rlxmm$~\cite{Margalit:2021:POPL,Norris:2013:OOPSLA} and causal consistency models~\cite{BouajjaniEGH17,Burckhardt:2014} which employ axioms which differ mildly from the release-acquire variants. 
From~\cite{ChakrabortyKMP24} we know that the analogous consistency testing problem is $\NP$-complete for all these models even when the number of threads, the number of memory locations and the domain of values are kept bounded. In this paper, we consider a different dimension to bound.

\noindent \textbf{Contributions.} We prove that the consistency testing problem under the release-acquire model and its variants can be solved in $\mathcal{O}(n^3)$ (where $n$ is the total number of memory operations in the instance) for \onewriter\ systems: these are programs where atmost one thread writes to a shared variable, but multiple threads can read from it. To the best of our knowledge, this is the first known polynomial-time algorithm for the consistency testing problem under the release-acquire semantics. 
Single-writer systems appear in several contexts, for instance, in producer-consumer models or CREW (concurrent-read-exclusive-write) distributed systems~\cite{Jaja1992ParallelAlgorithms}. Notably, for sequential consistency, the problem for single-writer systems is known to be $\NP$-hard~\cite{Gibbons}, and from~\cite{rv-equivalence} it is known to have a polynomial-time procedure when the number of threads is kept bounded. Here, we show polynomial-time under no restriction, for the release-acquire semantics.
As our second contribution, we present new  lower bounds for the consistency testing  problem under the release-acquire model and its variants: 
\begin{itemize}
\item For \onewriter\ programs, we prove that there is no $\mathcal{O}(n^{\frac{3}{2} - \epsilon})$-time algorithm, under the Boolean Matrix Multiplication (BMM) hypothesis, for any $\epsilon > 0$. 
\item When we consider \twowriter\ programs (those where each variable is written into by atmost two threads), the problem continues to be $\np$-complete. 
\item Furthermore, when we consider \threewriter\ programs, the problem admits a conditional lower bound under the Exponential Time Hypothesis (ETH): there is no $2^{o(k)} \cdot n^{\mathcal{O}(1)}$ algorithm, where $k$ is the number of threads, and $n$ is the total number of operations.
\end{itemize}
The results provide a tight complexity picture while bounding the number of writer threads per memory location, with \onewriter\ giving polynomial-time, \twowriter\ being $\NP$-hard and \threewriter\ allowing a stronger conditional lower bound.

\noindent \textbf{Structure of the paper.} In Section~\ref{sec:prelims}, we present the axiomatic definitions for the release-acquire semantics and its variants. 
In Section~\ref{sec:onewriter}, we present the polynomial-time procedure for single-writer systems and in Section~\ref{sec:hardness}, the hardness results. In the paper, we talk about the models $\ramm, \sramm$ and $\wramm$. The corresponding results for the closely related memory models $\rlxmm$ and the models for causal consistency are presented in Appendices~\ref{app:relaxed} and \ref{app:cc-models}. 
Finally, in Section~\ref{sec:conclusion}, we conclude the paper with some directions for future work. Missing proofs can be found in clearly marked appendices.

\section{Preliminaries}\label{sec:prelims}

\newcommand{\vm}{\mathrm{V}\MemModel}

Unlike SC, where consistency is checked via an interleaving, release-acquire semantics requires reasoning with two special types of relations on the reads and writes: a reads-from relation $\rf$ (mapping each read to the write it observes) and a per-location modification order $\mo$ (ordering writes to the same variable). The verification question becomes: does there exist $\rf$ and $\mo$ on the observed events that satisfy the axioms for the release-acquire memory model? The presentation in this section follows~\cite{Lahav2021,ChakrabortyKMP24} and introduces these preliminaries.

An \emph{event} is a memory operation: either a read $\rd(x, v)$ or write $\wt(x, v)$, where $x \in \varset$ is a shared memory location  and $v \in \valset$ is a value that is either read-from or written-to the memory location. These are the analogues of ``$read(x,v)$'' and ``$write(x,v)$'' that were used in~\cite{Gibbons} and recalled in the Introduction. We sometimes use the term \emph{variable} for a memory location. For an event $e$, we write $e.\op$, $e.\lloc$ and $e.\val$ respectively for the memory operation ($\rd$ or $\wt$), the shared memory location and the value associated with the event. Events are spread across a set of \emph{threads}.
For a set of events $\E$, we write $\E.\R$ and $\E.\W$ for the set of all read and write events in $\E$ respectively, and $\E.\R_x, \E.\W_x$ for the read and write events involving variable $x$. When the set of events $\E$ is clear from the context, we simply write $\R_x$ and $\W_x$.

\noindent \textbf{Execution graph.} An \emph{execution graph} \cite{Lahav:2017} (also known as candidate execution~\cite{Batty:2011})  
is a tuple $\ex = \tup{\E, \po, \rf, \mo}$  where $\E$ is a set of events and $\po$, $\rf$, $\mo$ are binary relations over $\E$. 
The \emph{program order} $\po$ is a partial order that enforces a total order on events of the same thread. 
The \emph{reads-from relation} $\rf\subseteq (\E.\W \times \E.\R)$ associates write
 events $\event_1$ to read events $\event_2$. It can be seen as an edge $e_1 \xra{~\rf~} e_2$ from (the node corresponding to) $e_1$ to (the node corresponding to) $e_2$ in the execution graph $\ex$. 
Indeed, it must be that $\event_1.\lloc=\event_2.\lloc$, $\event_1.\val=  \event_2.\val$ and every read event has a unique writer ($\rfinv$ is a function). 
The \emph{modification order} $\mo\subseteq (\E.\W\times \E.\W)$ is the union of modification orders $\mo_x$ for $x \in \varset$, where each $\mo_x$ is a total order over the set $\locx{\W}$ of write events that involve variable $x$. For a read $r$, let $\rf^{-1}(r)$ denote the write associated by $\rf$ to $r$, in other words, it is the write from which $r$ reads from according to $\rf$.

Consistency of an observed execution graph with respect to a weak memory model can be characterized using certain properties of the combination consisting of the $\po$, $\rf$ and $\mo$ relations -- these properties are called \emph{consistency axioms}. In this work, we focus on the release-acquire memory model ($\ramm$) and its close variants strong release-acquire ($\sramm$) and weak release-acquire ($\wramm$). In Section~\ref{subsec:consistency_axioms}, we formally define the consistency axioms. Before that, we state our problem of interest, the analogue of the VSC problem~\cite{Gibbons} for release-acquire semantics. The problem statement uses the notion of a \emph{partial execution graph} $\expartial=\tup{\E, \po}$, which is an abstraction of execution graphs without the modification order and reads-from relation.

\begin{center}
  \fbox{%
    \begin{minipage}{0.98\textwidth}
      \noindent \textsc{Verifying weak memory consistency} (The $\vm$-problem)

      \medskip

      \noindent \textbf{Instance:} A partial execution graph $\expartial=\tuple{\E, \po}$ and a memory model $\MemModel \in \{\ramm, \sramm, \wramm\}$

      \medskip

      \noindent \textbf{Question:} Do there exist $\rf$ and $\mo$ relations on $\E$ such that the complete execution graph $\ex = \tuple{\E, \po, \rf, \mo}$ satisfies the consistency axioms of the memory model $\MemModel$?
    \end{minipage}%
  }
\end{center}

\noindent{\bf{Notations for Relations}}. 
Let $S$ be a binary relation over the set of events $\E$.  The reflexive, transitive, reflexive-transitive closures, and inverse relations of $S$ 
are denoted as $S^?$, $S^+$, $S^*$, $S^{-1}$ respectively. 
The relation $S$ is \emph{acyclic} if $S^+$ is irreflexive. 
We write $\irr(S)$ and $\acy(S)$ to denote that  relation $S$ is irreflexive and acyclic respectively.  The identity relation on a set $B$ is denoted as $[B]$ : 
$[B](x,y) \triangleq x = y \land x \in B$.
The composition of two relations $S_1$ and $S_2$, is denoted by $S_1;S_2$. 
For a shared variable $x$, we let $S_x=[\locx{\E}];S;[\locx{\E}]$ be the restriction of $S$ to all events of $\E$ on $x$, with $\locx{\E} = \{ e \in E \mid e.\var = x \}$.

\noindent \textbf{Derived relations.} The \emph{happens-before relation} $\hb$ is the transitive closure of $\po$ and $\rf$. In our descriptions, we view $\hb$ as a path consisting of $\xra{\po}$ and $\xra{\rf}$ edges. Happens-before ($\hb$) projected to an individual variable $x$ is denoted as  $\hb_x$: this will be a path where all the $\rf$ edges are over the variable $x$. 
 More precisely: 
 $\hb \triangleq (\po\cup \rf)^+$, and $\hb_x \triangleq (\po\cup \rf_x)^+$.  We write $\rf_x$ for the projection of $\rf$ to variable $x$. 

 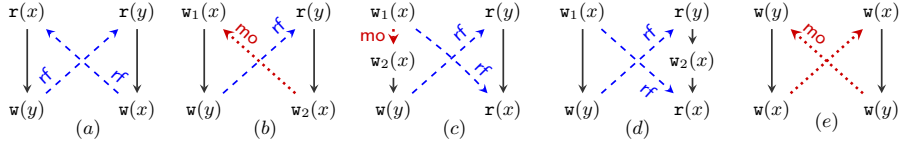
\begin{figure}[t]
    \def\ystep{0.4}
     \resizebox{\textwidth}{!}{
	\begin{tikzpicture}[yscale=1]
        \node (t11) at (0,0*\ystep)  {$\rd(x)$};
      \node (t12) at (0,-4*\ystep) {$\wt(y)$};
      \node (t21) at (1.8,0,0*\ystep) {$\rd(y)$};
      \node (t22) at (1.8,-4*\ystep) {$\wt(x)$};
     \node (t23) at (1,-4.8*\ystep) {$(a)$};
    
      \draw[po] (t11) to (t12);
      \draw[po] (t21) to (t22);
      \draw[rf,bend left=0] (t12) to node[above,pos=0.1, sloped]{$\rf$} (t21);
      \draw[rf,bend right=0] (t22) to node[above,pos=0.1, sloped]{$\rf$} (t11);

      \node (t11) at (2.9,0*\ystep)  {$\wt_1(x)$};
      \node (t12) at (2.9,-4*\ystep) {$\wt(y)$};
      \node (t21) at (4.7,0*\ystep) {$\rd(y)$};
      \node (t22) at (4.7,-4*\ystep) {$\wt_2(x)$};
    %
     \node (t23) at (3.9,-4.8*\ystep) {$(b)$};
    
      \draw[po] (t11) to (t12);
      \draw[po] (t21) to (t22);
    %
      \draw[rf,bend left=0] (t12) to node[above,pos=0.9, sloped]{$\rf$} (t21);
    %
      \draw[mo,bend left=0] (t22) to node[above,pos=0.8, sloped]{$\mo$} (t11);
    
      \node (t11) at (6,0*\ystep)  {$\wt_1(x)$};
      \node (t12) at (6,-2*\ystep) {$\wt_2(x)$};
      \node (t13) at (6,-4*\ystep) {$\wt(y)$};
      \node (t21) at (7.8,0*\ystep) {$\rd(y)$};
      \node (t22) at (7.8,-4*\ystep) {$\rd(x)$};
     \node (t23) at (7,-4.8*\ystep) {$(c)$};
    
      \draw[po] (t12) to (t13);
      \draw[po] (t21) to (t22);
     \draw[mo] (t11) to node[left]{$\mo$} (t12);
      \draw[rf,bend left=0] (t13) to node[above,pos=0.8, sloped]{$\rf$} (t21);
      \draw[rf,bend left=5] (t11) to node[above,pos=0.8, sloped]{$\rf$} (t22);
    %
      \node (t11) at (9.1,0*\ystep)  {$\wt_1(x)$};
      \node (t12) at (9.1,-4*\ystep) {$\wt(y)$};
      \node (t21) at (10.9,0*\ystep) {$\rd(y)$};
      \node (t22) at (10.9,-2*\ystep) {$\wt_2(x)$};
      \node (t23) at (10.9,-4*\ystep) {$\rd(x)$};
    \node (t24) at (10,-4.8*\ystep) {$(d)$};
      \draw[po] (t11) to (t12);
      \draw[po] (t21) to (t22);
      \draw[po] (t22) to (t23);
      \draw[rf,bend left=0] (t12) to node[above,pos=0.8, sloped]{$\rf$} (t21);
      \draw[rf,bend right=0] (t11) to node[below,pos=0.8, sloped]{$\rf$} (t23);
    
     \node (t11) at (12.2,0*\ystep)  {$\wt(y)$};
      \node (t12) at (12.2,-4*\ystep) {$\wt(x)$};
    %
      \node (t21) at (14,0*\ystep) {$\wt(x)$};
      \node (t22) at (14,-4*\ystep) {$\wt(y)$};
      \node (t23) at (13.1,-4.5*\ystep) {$(e)$};
      \draw[po] (t11) to (t12);
    %
      \draw[po] (t21) to (t22);
    %
    %
      \draw[mo,bend left=0] (t12) to (t21);
      \draw[mo,bend right=0] (t22) to node[above,pos=0.8, sloped]{$\mo$} (t11);
    
    \end{tikzpicture}
	 }
    \caption{Violations of coherence and causality. Values in the $\rd$ and $\wt$ operations suppressed for clarity. Black edge denotes $\po$: (a)~\PORF{} is violated (b)~\WCoh{} is violated: $\wt_1(x)~\hb~\wt_2(x)$ while 
    $\wt_2(x)~\mo~\wt_1(x)$, 
     (c)~\RCoh{} is violated : $\wt_1(x)~\rf~\rd(x)$ with $\wt_1(x)~\mo~\wt_2(x)$ and $\wt_2(x)~\hb~\rd(x)$, 
       (d)~\WRCoh{} is violated : $\wt_1(x)~\rf~\rd(x)$ with $\wt_1(x)~\hb~\wt_2(x)$
      and $\wt_2(x)~\hb~\rd(x)$,    (e)~\SWCoh{} is violated : $\hb \cup \mo$ cycle. }
      \label{fig:illust}
\end{figure}

 \subsection{Consistency Axioms}\label{subsec:consistency_axioms}
The axioms are broadly classified as  
\emph{coherence} and \emph{causality cycles}. Table~\ref{table:ax:models-2} states the  axioms and indicates the axioms for each memory model. Figure~\ref{fig:illust} gives scenarios for violation of the various axioms described below.

\noindent{\bf{Coherence}}. The modification order $\mo$ enforces an `SC-per-variable' constraint in any execution $\ex$, that is, memory accesses per shared variable are  totally ordered.  Coherence axioms essentially describe rules on the $\mo$.
 
\emph{Write-coherence} enforces that $\mo$ agrees with $\hb$ for each variable (denoted as $\irr(\mo_x;\hb)$ in Table~\ref{table:ax:models-2} and a violation of it is shown in Figure~\ref{fig:illust}(b)): intuitively, if a write $\wt(x, v_1)$ happens-before $\wt(x, v_2)$, then $x$ is updated to $v_1$ first, and then $v_2$. As the order of updates to $x$ is encoded by the modification order $\mo_x$, the axiom requires $\mo_x$ to align with $\hb$.  
A stronger version of it, called \emph{strong-write-coherence}, requires the entire $\mo$ (and not just $\mo_x$) to align with $\hb$, in other words, it requires the acyclicity of $\po \cup \rf \cup \mo$ -- again, notice that this requirement is across variables, as $\mo=\bigcup_{x \in \varset}\mo_x$. Execution graph in Figure~\ref{fig:illust}(e) violates strong-write-coherence, while it does not violate write-coherence.

 \emph{Read coherence}  enforces that a read $\rd(x)$ cannot read from a write $\wt_1(x)$ if 
there is an \emph{intermediate} write $\wt_2(x)$ that happens-before $\rd(x)$, 
i.e. $(\wt_2(x), \rd(x)) \in \hb$ (denoted as $\wt_2(x) \xra{\hb} \rd(x)$), and ``intermediate'' refers to $\mo$,
i.e.,  $(\wt_1(x),\wt_2(x))\in \mo$ (denoted as $\wt_1(x) \xra{\mo} \wt_2(x)$). Said differently, if $\wt_1(x)$ and $\wt_2(x)$ happen-before a read $\rd(x)$, then $\wt_1 \xra{\mo} \wt_2$, implies $\rd(x)$ cannot read from $\wt_1(x)$. The $\mo$ relation makes $\wt_1$ stale for $\rd$. Figure~\ref{fig:illust}(c) shows a violation of read-coherence. 
In a weaker form of read-coherence called
\emph{weak-read-coherence}~\cite{Lahav:2022} (Figure~\ref{fig:illust}(d)), 
``intermediate'' relates to $\hb$ (and not $\mo$), i.e., we have $\wt_1(x) \xra{\hb} \wt_2(x)$: in other words, if two writes $\wt_1 \xra{\hb} \wt_2 \xra{\hb} \rd$, then $\rd$ cannot read from $\wt_1$. 
 
\noindent{\bf{Causality cycles}}. A causality cycle consists of $\po$ and $\rf$ orderings.
The C11 models explicitly mandate \PORF{} \cite{Lahav:2022,Luo:2021}, that is, 
$\acy(\po \cup \rf)$ (Figure~\ref{fig:illust}(a)). 

\begin{table}[t!]
     \caption{\small Models $\wramm, \ramm, \sramm$.}
    \label{table:ax:models-2}  
        \setlength\tabcolsep{7.5pt}
      \begin{tabular}{lr|lr}
      \multicolumn{2}{l|}{} & \multicolumn{2}{l}{} \\
       &  & $\irr(\mo_x;\hb)$  (\WCoh{})~~ & \\
     $\wramm$ & \makecell{\PORF{}\\ \WRCoh{}} &$\irr(\rf^{-1};\mo_x;\hb)$  (\RCoh{})&\\  
      &&$\acy(\hb \cup \mo)$  (\SWCoh{})~~& \\
       && $\irr(\hb;[\W_x];\hb;\rf^{-1}_x)~~~$ (\WRCoh{})& \\
      \hline
      $\ramm$ & \makecell{\PORF{} \\
      \WCoh{}\\ \RCoh{}}   &$\sramm$  \makecell{\PORF{}\\ \SWCoh{}\\ \RCoh{}} &\\ 
      \cline{1-4}
      \end{tabular}
          
  \end{table}

An execution graph $\ex = (\E, \po, \rf, \mo)$ is consistent with memory model $\MemModel$ if it satisfies all the axioms corresponding to $\MemModel$ as shown in Table  ~\ref{table:ax:models-2}. We write $\ex \models \MemModel$ if $\ex$ is consistent with $\MemModel$. For a partial execution graph $\expartial = \tuple{\E, \po}$, we write $\expartial \models \MemModel$ if there exists an $\rf$ and $\mo$ such that $\ex \models \MemModel$ for the complete execution graph $\ex = \tuple{\E, \po, \rf, \mo}$. Hence, the $\vm$ problem can be stated as: given a partial execution graph $\expartial$ and a memory model $\MemModel$, does $\expartial \models \MemModel$? 

We conclude this section with a useful lemma describing the relationship between these three models. 
\begin{lemma}\label{lem:sra-ra-wra-relationship}
Let $\ex = \tuple{\E, \po, \rf, \mo}$ be an execution graph. Then: $\ex \models \sramm$ implies $\ex \models \ramm$, and $\ex \models \ramm$ implies $\ex \models \wramm$. 
\end{lemma} 
\begin{proof}
  It is easy to see that if $\ex$ satisfies strong-write-coherence, then it satisfies write-coherence also. Hence $\ex \models \sramm$ implies $\ex \models \ramm$. For $\ramm$ to $\wramm$, we show that if weak-read-coherence is violated, then either read-coherence or write-coherence is violated: suppose weak-read-coherence is violated, there is a cycle $\rd(x) \xra{\rf^{-1}} \wt_1(x) \xra{\hb} \wt_2(x) \xra{\hb} \rd(x)$. The $\mo$ relation can order $\wt_1$ and $\wt_2$ in either way. If $\wt_1(x) \xra{\mo} \wt_2(x)$, then read-coherence is violated, and if $\wt_2(x) \xra{\mo} \wt_2(x)$, then write-coherence is violated.
\end{proof}

\section{Polynomial-time algorithm for \onewriter\ systems}
\label{sec:onewriter}

We now consider programs where, for every variable, there is a single thread which writes to it. Any number of threads can read from a variable. An execution graph $\ex = (\E, \po, \rf, \mo)$ is said to be a \onewriter\ execution graph if in the set of events $\E$, there is atmost one thread writing to every memory location. 

\noindent \textbf{Overview of this section.} The fact that we work with \onewriter\ execution graphs simplifies several axioms of Table~\ref{table:ax:models-2}, allowing for  a polynomial-time procedure. Recall however that for sequential consistency, the \onewriter\ assumption offers no specific advantage and still yields a $\NP$-hard complexity~\cite{Gibbons}. 

\noindent \textit{$\sramm$ and $\ramm$ reduce to $\wramm$.} Firstly, we note that consistency checking for the $\sramm$ and $\ramm$ memory models reduce to checking for $\wramm$ in \onewriter\ execution graphs (Lemma~\ref{lem:one-writer-reduces-to-wra}). Moreover, since all writes of a memory location lie in a single thread, the $\mo$ relation on writes gets fixed to the unique order that agrees with $\po$.  Therefore, the problem boils down to detecting if there is an $\rf$ that satisfies \emph{porf-acyclicity} and \emph{weak-read-coherence}, the two axioms needed for $\wramm$ consistency. So we focus on these two axioms for the rest of the section. 

\noindent \textit{A partial order on $\rf$.} Secondly, since all writes to a memory location $x$ are totally ordered by $\po$, there is a natural partial order induced on the $\rf$ relations. A reads-from relation $\rf_1$ is said to be smaller than $\rf_2$, written as $\rf_1 \sqsubseteq \rf_2$, if for every 
read $r$, the write event $w_2 = \rf^{-1}_2(r)$ appears $\po$-later than $w_1 = \rf^{-1}_1(r)$ (we say $w_2$ is $\po$-later than $w_1$ if $w_1~\po^*w_2$). In other words, $\rf_2$ associates each read $r$ to a write that appears later in the program order compared to the write that $\rf_1$ associates it to\label{order}. 
We observe two useful properties of this partial order: (1) if an $\rf$ violates porf-acyclicity, then every $\rf^{\dagger}$ larger than it also violates porf-acyclicity, and (2) among all the $\rf$ that satisfy weak-read-coherence, there is a unique minimum element $\rf_{\min}$.

\noindent \textit{Algorithm.} Using these two properties, we design an algorithm that starts with the smallest $\rf$ and iteratively computes larger reads-from relations until a $\wramm$-consistent $\rf$ is found, or concludes that no such $\rf$ exists.  

The first two ideas appear in Section~\ref{sec:onewriter-structural-properties} and the algorithm is described in Section~\ref{sec:onewriter-algorithm}.

\subsection{Structural properties of \onewriter\ execution graphs}
\label{sec:onewriter-structural-properties}
Our first observation is that for \onewriter\ systems, checking $\sramm$ and $\ramm$ reduces to checking $\wramm$.  For $\wramm$ the only axioms are porf-acyclicity and weak-read-coherence.  For $\ramm$ one must also ensure write-coherence and read-coherence, in addition to porf-acyclicity.  In a \onewriter\ setting, write-coherence is enforced by choosing $\mo$ to agree with $\po$ (so all writes to a location are ordered by the program order of their single writer).  

A violating cycle for read-coherence has the form
\[
r_x \xra{\rf^{-1}} w_x \xra{\mo_x} w'_x \xra{\hb} r_x.
\]
When $\mo$ agrees with $\po$, the $\mo_x$ edge is simply $\po^+$, since both $w_x$ and $w'_x$ are in the same thread. Now, a violating cycle for weak-read-coherence has the form
\[
r_x \xra{\rf^{-1}} w_x \xra{\hb} w'_x \xra{\hb} r_x,
\]
and, since $w_x$ and $w'_x$ come from the same thread, the $\hb$ step between $w_x$ and $w'_x$ can likewise be instantiated by $\po^+$.  Hence read-coherence coincides with weak-read-coherence in the \onewriter\ setting.  Moreover, having an $\mo$ agreeing with $\po$ reduces strong-write-coherence to porf-acyclicity, culminating in this lemma.  

\begin{restatable}{lemma}{lemonewriterreducestowra}
    \label{lem:one-writer-reduces-to-wra}
 Let $\ex = \tuple{\E, \po, \rf, \mo}$ be a \onewriter\ execution graph. Then, for $\MemModel \in \{\sramm, \ramm\}$, we have $\ex \models \MemModel$ iff $\mo$ agrees with $\po$ and $\ex \models \wramm$. 
\end{restatable}

As a consequence, the consistency checking of \onewriter\ under $\sramm$, $\ramm$ and $\wramm$ boils down to checking porf-acyclicity and weak-read-coherence violation.
Secondly, the problem of synthesizing both $\rf$ and $\mo$ boils down to just synthesizing $\rf$, since $\mo$ is fixed to the order that respects $\po$. Nevertheless, there are still exponentially many $\rf$ functions possible. The next lemma entails that there is a ``smallest'' $\rf$ among all $\wramm$-consistent $\rf$s in the $\sqsubseteq$ order on $\rf$s defined in the overview of this section. Given $\rf_1, \rf_2$, we write $\min(\rf_1, \rf_2)$ for the reads-from relation which associates every read $r$ to the $\po$-smaller write between $\rf^{-1}_1(r)$ and $\rf^{-1}_2(r)$.

\begin{restatable}{lemma}{lemmaonewritertwoproperties}
    \label{lem:onewriter-axioms-two-properties}
Let $\rf_1$ and $\rf_2$ be two reads-from relations over a partial \onewriter\ execution graph $\ex=\tuple{\E, \po, \mo}$ s.t. $\mo$ agrees with $\po$. 
\begin{itemize}
\item Suppose $\rf_1 \sqsubseteq \rf_2$. Then, if $\rf_1$ induces a $\po \cup \rf_1$ cycle, then $\rf_2$ induces a $\po \cup \rf_2$ cycle.
\item Suppose both $\rf_1$ and $\rf_2$ satisfy weak-read-coherence. Then, the $\rf$ function $\min(\rf_1, \rf_2)$ also satisfies weak-read-coherence.
\end{itemize} 
\end{restatable}

Suppose $S$ is the set of all $\rf$ that satisfy weak-read-coherence. The second item of Lemma~\ref{lem:onewriter-axioms-two-properties} implies that $S$ contains a unique minimum element, which we denote by $\rf_{\min}$. By the first item of Lemma~\ref{lem:onewriter-axioms-two-properties}, if $\rf_{\min}$ violates porf-acyclicity then every $\rf$ in $S$ also violates it. Therefore, $\wramm$-consistency checking reduces to computing $\rf_{\min}$ and testing it for porf-acyclicity. The algorithm in the next section performs exactly this construction and test. 

\begin{figure}[!tbp]
   \def\ystep{0.5}

\centering
	\scalebox{.7}{
\begin{tikzpicture}[yscale=1]
\node (t01) at (-.2, 1.5*\ystep) {$t_1$};
\node (t02) at (1.5, 1.5*\ystep) {$t_2$};
\node (t03) at (3, 1.5*\ystep) {$t_3$};

  \node (t11) at (-.2, 0*\ystep) {$\wt_1(x,1)$};
    \node (t12) at (1.5, 0*\ystep) {$\rd_1(x,2)$};
    \node (t13) at (3, 0*\ystep) {$\rd_2(y,1)$};

    \node (t21) at (-.2, -2*\ystep) {$\wt_2(x,2)$};
    \node (t22) at (1.5, -2*\ystep) {$\rd_3(x,1)$};    
    
    \node (t23) at (3, -2*\ystep) {$\rd_4(x,2)$};

    \node (t31) at (-.2, -4*\ystep) {$\wt_3(x,1)$};
    \node (t32) at (1.5, -4*\ystep) {$\wt_4(y,1)$};
    
    \node (t41) at (-.2, -6*\ystep) {$\wt_5(x,2)$};
    
    \node (n1) at (1.5,-8*\ystep) {$\rf_0$};
    
    \draw[po] (t11) -- (t21);
    \draw[po] (t21) -- (t31);
    \draw[po] (t31) -- (t41);
    
    \draw[po] (t12) -- (t22);
    \draw[po] (t22) -- (t32);
    \draw[po] (t13) -- (t23);

    \draw[rf] (t11) -- (t22);
    \draw[rf] (t21) -- (t12);
    \draw[rf] (t21) to[out=-25,in=-160] (t23);
    \draw[rf] (t32) -- (t13);
    
    \node (t01) at (5, 1.5*\ystep) {$t_1$};
\node (t02) at (6.5, 1.5*\ystep) {$t_2$};
\node (t03) at (8, 1.5*\ystep) {$t_3$};

  \node (t11) at (5, 0*\ystep) {$\wt_1(x,1)$};
    \node (t12) at (6.5, 0*\ystep) {$\rd_1(x,2)$};
    \node (t13) at (8, 0*\ystep) {$\rd_2(y,1)$};

    \node (t21) at (5, -2*\ystep) {$\wt_2(x,2)$};
    \node (t22) at (6.5, -2*\ystep) {$\rd_3(x,1)$};   
    \node (t23) at (8, -2*\ystep) {$\rd_4(x,2)$};

    \node (t31) at (5, -4*\ystep) {$\wt_3(x,1)$};
    \node (t32) at (6.5, -4*\ystep) {$\wt_4(y,1)$};
    
    \node (t41) at (5, -6*\ystep) {$\wt_5(x,2)$};
    
    \node (n1) at (6.5,-8*\ystep) {$\rf_1$};
    
    \draw[po] (t11) -- (t21);
    \draw[po] (t21) -- (t31);
    \draw[po] (t31) -- (t41);
    
    \draw[po] (t12) -- (t22);
    \draw[po] (t22) -- (t32);
    \draw[po] (t13) -- (t23);

    \draw[rf] (t31) -- (t22);
    \draw[rf] (t21) -- (t12);
    \draw[rf] (t21) to[out=-25,in=-160] (t23);
    \draw[rf] (t32) -- (t13);
    
    
    \node (t01) at (10, 1.5*\ystep) {$t_1$};
\node (t02) at (11.5, 1.5*\ystep) {$t_2$};
\node (t03) at (13, 1.5*\ystep) {$t_3$};

  \node (t11) at (10, 0*\ystep) {$\wt_1(x,1)$};
    \node (t12) at (11.5, 0*\ystep) {$\rd_1(x,2)$};
    \node (t13) at (13, 0*\ystep) {$\rd_2(y,1)$};

    \node (t21) at (10, -2*\ystep) {$\wt_2(x,2)$};
    \node (t22) at (11.5, -2*\ystep) {$\rd_3(x,1)$};   
    \node (t23) at (13, -2*\ystep) {$\rd_4(x,2)$};

    \node (t31) at (10, -4*\ystep) {$\wt_3(x,1)$};
    \node (t32) at (11.5, -4*\ystep) {$\wt_4(y,1)$};
    
    \node (t41) at (10, -6*\ystep) {$\wt_5(x,2)$};
    
    \node (n1) at (11.5,-8*\ystep) {$\rf_2$};
    
    \draw[po] (t11) -- (t21);
    \draw[po] (t21) -- (t31);
    \draw[po] (t31) -- (t41);
    
    \draw[po] (t12) -- (t22);
    \draw[po] (t22) -- (t32);
    \draw[po] (t13) -- (t23);

    \draw[rf] (t31) -- (t22);
    \draw[rf] (t21) -- (t12);
    \draw[rf] (t41) to[out=-15,in=250] (t23);
    \draw[rf] (t32) -- (t13);

     \end{tikzpicture}}

\caption{\footnotesize From left: $\rf_0 \sqsubseteq \rf_1 \sqsubseteq \rf_2$. $\rf_0$ has a violation of \WRCoh{}:~$\rd_3 \xra{\rf_{0}^{-1}} \wt_1 \xra{\po} \wt_2 \xra{\hb_0} \rd_3 $.  $\rf_1$ has a violation of \WRCoh{}:~$\rd_4 \xra{\rf_{1}^{-1}} \wt_2 \xra{\po} \wt_3 \xra{\hb_1} \rd_4 $. $\rf_2$ does not have any more violations of \WRCoh{}. \PORF{} is also satisfied by $\rf_2$. The underlying partial execution graph is $\wramm$-consistent.
}\label{Example-one-writer-without-porfcycle}

\end{figure}

\subsection{Algorithm}
\label{sec:onewriter-algorithm}

Our algorithm starts from an $\rf_0$ which maps every read $r$ to the $\po$-smallest  write in its writer thread, such that the value of the chosen write matches with that of $r$, and in case the writer thread is the same as the one where $r$ is present, the mapped write lies above $r$. This relation $\rf_0$ may not satisfy weak-read-coherence. In that case, a violating cycle is identified and $\rf_0$ is modified to $\rf_1$ as explained below. Assume we have iteratively constructed upto $\rf_i$, and $\rf_i$ does not satisfy weak-read-coherence.

\begin{wrapfigure}[20]{r}{0.4\textwidth}
   \def\ystep{0.5}

\centering
	\scalebox{.7}{
\begin{tikzpicture}[yscale=1]
\node (t01) at (-.2, 1.5*\ystep) {$t_1$};
\node (t02) at (1.5, 1.5*\ystep) {$t_2$};

  \node (t11) at (-.2, 0*\ystep) {$\wt_1(x,1)$};
    \node (t12) at (1.5, 0*\ystep) {$\rd_1(x,2)$};
    
    \node (t21) at (-.2, -2*\ystep) {$\wt_2(x,2)$};
    \node (t22) at (1.5, -2*\ystep) {$\rd_3(x,1)$};    
    
    \node (t31) at (-.2, -4*\ystep) {$\rd_2(y,1)$};
    \node (t32) at (1.5, -4*\ystep) {$\wt_4(y,1)$};
    
    \node (t41) at (-.2, -6*\ystep) {$\wt_5(x,1)$};
    
    \node (n1) at (1.1,-8*\ystep) {$\rf_0$};
    
    \draw[po] (t11) -- (t21);
    \draw[po] (t21) -- (t31);
    \draw[po] (t31) -- (t41);
    
    \draw[po] (t12) -- (t22);
    \draw[po] (t22) -- (t32);
      
    \draw[rf] (t21) to (t12);
    \draw[rf] (t11) -- (t22);
    \draw[rf] (t32) to (t31);
    \node (t01) at (4, 1.5*\ystep) {$t_1$};
\node (t02) at (5.5, 1.5*\ystep) {$t_2$};

  \node (t11) at (4, 0*\ystep) {$\wt_1(x,1)$};
    \node (t12) at (5.5, 0*\ystep) {$\rd_1(x,2)$};
    
    \node (t21) at (4, -2*\ystep) {$\wt_2(x,2)$};
    \node (t22) at (5.5, -2*\ystep) {$\rd_3(x,1)$};   
    
    \node (t31) at (4, -4*\ystep) {$\rd_2(y,1)$};
    \node (t32) at (5.5, -4*\ystep) {$\wt_4(y,1)$};
    
    \node (t41) at (4, -6*\ystep) {$\wt_5(x,1)$};
    
    \node (n1) at (5.2,-8*\ystep) {$\rf_1$};
    
     \draw[po] (t11) -- (t21);
    \draw[po] (t21) -- (t31);
    \draw[po] (t31) -- (t41);
    
    \draw[po] (t12) -- (t22);
    \draw[po] (t22) -- (t32);
    
    \draw[rf] (t21) to (t12);
    \draw[rf] (t41) to[out=30,in=-135] (t22);
    \draw[rf] (t32) to (t31);
    
    \end{tikzpicture}}

\caption{\footnotesize From left: $\rf_0 \sqsubseteq \rf_1$. $\rf_0$ has a violation of \WRCoh{}:~$\rd_3 \xra{\rf_{0}^{-1}} \wt_1 \xra{\po} \wt_2 \xra{\hb} \rd_3 $.  $\rf_1$ has no \WRCoh{} violation but does not satisfy \PORF{}:~$\wt_5 \xra{\rf_{1}} \rd_3 \xra{\po} \wt_4 \xra{\rf} \rd_2 \xra{\po} \wt_5$. The partial execution graph is not $\wramm$-consistent.
}\label{Example-one-writer-with-porfcycle}  
\end{wrapfigure}
\noindent \emph{From $\rf_i$ to $\rf_{i+1}$.} Pick a violating cycle for weak-read-coherence:
\[ r_i \xra{\rf_{i}^{-1}} w_i \xra{\hb_i} w'_i \xra{\hb_i} r_i \]
where all three events $r_i, w_i, w'_i$ are over the same variable, say $x_i$, and $\hb_i$ is $(\po \cup \rf_i)^+$. Recall that since we are in the \onewriter\ case, the intermediate $\hb_i$ can be replaced with $\po^+$. Modify $\rf_i$ to eliminate the violating cycle: choose a write for $r_i$ which appears on or after $w'_i$, that is, pick $w''_i$ such that $w'_i~\po^* w''_i$ and $w''_i$ is a write on the same variable $x_i$, writing the same value as the read event $r_i$. If no such write exists, conclude that there is no $\rf$ for $\expartial$ that satisfies weak-read-coherence. Else, let the resulting $\rf$ be called $\rf_{i+1}$. If $\rf_{i+1}$ does not satisfy weak-read-coherence, we update $\rf_{i+1}$ as above. Else, $\rf_{i+1}$ satisfies weak-read-coherence. We declare $\rf_{\min} = \rf_{i+1}$.

Finally, we test whether $\rf_{\min}$ satisfies porf-acyclicity. If it does, we have a witness for $\wramm$-consistency. Else, we conclude (due to Lemma~\ref{lem:onewriter-axioms-two-properties}) that $\expartial$ is not $\wramm$-consistent. 
We formalize this discussion in Algorithm~\ref{algo:one-writer-wra} in Appendix~\ref{app:algorithm}. Here is the key invariant maintained by the algorithm: at each step $i$ where the update function is called on $\rf_i$,  we can show that $\rf_i \sqsubseteq \rf^{\dagger}$ for every $\rf^\dagger$ that satisfies weak-read-coherence. Examples~\ref{Example-one-writer-without-porfcycle} and ~\ref{Example-one-writer-with-porfcycle} illustrate the running of the algorithm on two different partial execution graphs. 

\begin{restatable}{theorem}{thmwracorrectness}\label{th-wra-correctness}
A partial execution graph $\expartial = \tuple{\E, \po}$ of a one-writer system is consistent for $\wramm$ iff the final $\rf$ computed by the procedure discussed above satisfies weak-read-coherence and porf-acyclicity. Moreover, the algorithm runs in time $\mathcal{O}(n^3)$ where $n = |E|$ and returns the smallest $\rf$  that is $\wramm$-consistent, if one exists.
\end{restatable}

The major contributor to the complexity of Algorithm~\ref{algo:one-writer-wra} is verification of whether a given $\rf$ satisfies weak-read-coherence. This can be posed as a graph problem: there are $n$ nodes corresponding to the events, and there are three kinds of edges $\po$, $\rf$ and $\rf^{-1}$; does there exist a cycle of the form $r \xra{\rf^{-1}} w \xra{\po^+} w' \xra{(\po \cup \rf)^*} r$?  For each read event $r$ and its corresponding write event $w$, we need to check if there exist a $w'$ such that (1) there is a path of $\po$ edges from $w$ to $w'$, and (2) there is a path from $w'$ to $r$ over a combination of $\po$ and $\rf$ edges. For each read, this can be achieved by in $\mathcal{O}(n)$ time by a graph search. Hence, overall, across all possible reads, this algorithm takes $\mathcal{O}(n^2)$ time. Therefore, each update to an $\rf$ costs $\mathcal{O}(n^2)$ and there can be at most $n$ updates (as each update moves the write of a read to a $\po$-later write). This gives a complexity of $\mathcal{O}(n^3)$.  As our next observation, we prove a conditional lower bound of $\mathcal{O}(n^{\frac{3}{2} - \epsilon})$ on the time to check for weak-read-coherence in a one-writer program with $n$ events. 
This also transfers as a lower bound for the $\wramm$-consistency checking. Due to Lemma~\ref{lem:one-writer-reduces-to-wra}, it also gives a lower bound on the consistency checking for $\sramm$ and $\ramm$ as well. To show the lower bound, we make use of the Boolean Matrix Multiplication hypothesis and a lower bound result on finding triangles in graphs. 

\paragraph*{Boolean Matrix Multiplication (BMM) hypothesis.} There is no $\mathcal{O}(n^{3-\epsilon})$ algorithm for Boolean Matrix Multiplication, for any $\epsilon > 0$.

\begin{theorem}[Theorem~1.3 of \cite{WilliamsW18}]
  Under BMM hypothesis, there is no $\mathcal{O}(n^{3 - \epsilon})$ algorithm for detecting a triangle in an (undirected) graph with $n$ vertices, for any $\epsilon > 0$. 
\end{theorem}

We reduce triangle detection in a graph with $n$ vertices to weak-read-coherence checking in a \onewriter\ execution graph with $n^2$ events. This results in the following proposition. 

\begin{proposition}
    \label{prop:wra-onewriter-lowerbound}
Under BMM, for any $\epsilon > 0$, there is no $\mathcal{O}(n^{\frac{3}{2} - \epsilon})$ algorithm for deciding if a given \onewriter\ execution graph $\ex = \tuple{\E, \po, \rf, \mo}$ satisfies the weak-read-coherence axiom. 
\end{proposition}

Given graph $G = (V_G,\;E_G)$ with $n$ vertices, we build an execution graph $\expartial_G = \tuple{\E, \po}$ of a single-writer program with $n^2$ events, such that $\expartial$ is consistent w.r.t. $\MemModel$ for iff $G$ has no triangles. Therefore, if the consistency problem can be solved in $\mathcal{O}(n^{\frac{3}{2}\;-\;\epsilon})$ for some $\epsilon$, 
we can solve triangle detection on $n$-node graphs in $\mathcal{O}((n^2)^{\frac{3}{2}\;-\;\epsilon}) = \mathcal{O}(n^{3 - \frac{\epsilon}{2}})$,  contradicting the BMM hypothesis.

\begin{wrapfigure}[]{r}{0.55\textwidth}
  \centering
  \resizebox{0.5\textwidth}{!}{
  \begin{tikzpicture}
    \node (x1) at (0,0) {\scriptsize $x$};
    \node (y) at (2, 0) {\scriptsize $y$};
    \node (z) at (4, 0) {\scriptsize $z$};
    \node (x2) at (6, 0) {\scriptsize $x$};
    \draw [thin] (x1) to (y) to (z) to (x2);

    \begin{scope}[yshift=-0.6cm]
    \node (a) at (0,0) {\scriptsize $t_{a_x}$};
    \node (b) at (2, 0) {\scriptsize $t_{b_y}$};
    \node (c) at (4, 0) {\scriptsize $t_{c_z}$};
    \node (d) at (6, 0) {\scriptsize $t_{d_x}$};

    \node [right] at (-0.5, -0.5) {\scriptsize $\wt(a_x, 0)$};
    \node [right] at (-0.5, -0.8) {\scriptsize $\wt(a_x, 1)$};
    \node [right] (1) at (-0.5, -1.1) {\scriptsize $\wt(a_{xy}, 0)$};

    \node [right] (2) at (1.5, -0.5) {\scriptsize $\rd(a_{xy}, 0)$};
    \node [right] (3) at (1.5, -1.1) {\scriptsize $\wt(b_{yz}, 0)$};

    \draw [blue, ->, >=stealth] (1) to  (2);

    \node [right] (4) at (3.5, -0.5) {\scriptsize $\rd(b_{yz}, 0)$};
    \node [right] (5) at (3.5, -1.1) {\scriptsize $\wt(c_{zx}, 0)$};
    \draw [blue, ->, >=stealth] (3) to  (4);

    \node [right] (6) at (5.5, -0.5) {\scriptsize $\rd(c_{zx}, 0)$};
    \node [right] (7) at (5.5, -1.1) {\scriptsize $\rd(a_x, 0)$};
    \draw [blue, ->, >=stealth] (5) to  (6);

    \end{scope}
  \end{tikzpicture}
  }  
  \caption{From triangles to violations of weak-read-coherence} 
  \label{fig:super-linear-illustration}
\end{wrapfigure}
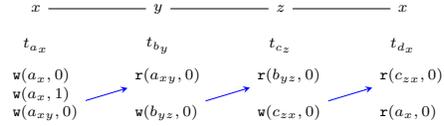
Our idea for the reduction is as follows. We will construct $\expartial$ in such a way that for every $\rd(x, v)$ there is a unique $\wt(x, v)$ present in the entire program. Therefore, the $\rf$ relation is fixed. If $G$ has a triangle, we will have a shared variable $a$ and a cycle of the form: $\wt(a, 0) \xrightarrow{\po} \wt(a, 1) \xrightarrow{\hb} \rd(a, 0) \xrightarrow{\rf^{-1}} \wt(a, 0)$. This cycle violates weak-read-coherence.
If $\expartial$ with the fixed $\rf$ violates  
\WRCoh{}, then $G$ has a triangle.  The $\hb$ path $\wt(a, 1) \xrightarrow{\hb} \rd(a, 0)$ should reflect the presence of a triangle in $G$. 

The key aspect of the construction is depicted in Figure~\ref{fig:super-linear-illustration}. Suppose $G$ contains a triangle over vertices $x, y, z$. We consider a variable $a_x$ for which we extract an $\hb$ path as above. To force the $\hb$ path to have length $3$, we consider four copies $t_{a_x}, t_{b_x}, t_{c_x}, t_{d_x}$ for each vertex $x$. When there is edge $(x, y)$, thread $t_{a_x}$ writes to $t_{b_y}$. Similarly, $(y, z)$ makes $t_{b_y}$ write to $t_{c_z}$ and finally the edge $(z, x)$ makes $t_{c_z}$ write to $t_{d_x}$ which contains $\rd(a_x, 0)$. The full construction is explained in Appendix~\ref{app:super-linear-lower-bound}.

\section{Hardness in the presence of multiple writers}
\label{sec:hardness}

While having a single writer per location induced a polynomial-time algorithm for the consistency check, we show that having two writers per thread makes the problem $\NP$-complete. The $\NP$ upper bound follows by guessing an $\rf$ and $\mo$ and then checking whether the complete execution graph $\ex = \tuple{\E, \po, \rf, \mo}$ satisfies the consistency axioms. When $\rf$ and $\mo$ are already given, this test can be done in polynomial-time~\cite{Tunc2023}. Hence the $\NP$ upper bound follows. The hardness for our $\vm$ problem comes from the fact that we need to \emph{synthesize} suitable $\rf$ and $\mo$ (the analogue in the sequential consistency setting would be the synthesis of an interleaving). For proving hardness, we provide a reduction from 3-CNF-SAT: given a boolean formula $\varphi$ over $k$ variables and $m$ clauses, with each clause being a disjunction of three literals, is there a satisfying assignment for $\varphi$?  We provide two reductions.
\begin{enumerate}
\item Formula $\varphi$ to a partial execution graph $\expartial_\varphi$ with $\mathcal{O}(k)$ threads, and containing at most \emph{three threads} writing to a location (\threewriter). 
\item Formula $\varphi$ to a partial execution graph $\ypartial_\varphi$ with $\mathcal{O}(k + m)$ threads, and containing at most \emph{two threads} writing to a location (\twowriter). 
\end{enumerate}

In both the reductions, we achieve the property that $\varphi$ is satisfiable iff $\expartial_\varphi$ and $\ypartial_\varphi$ are $\MemModel$-consistent, for $\MemModel \in \{\sramm, \ramm, \wramm\}$. This shows that the consistency checking problem is $\NP$-complete with a bounded number of writers per location, in fact, even with two writers. The reason we present the first reduction to \threewriter\ is because we get a stronger conditional bound: since the number of threads in $\expartial_\varphi$ is linearly related to the number of variables in $\varphi$, we get a parameterized reduction that helps proving a conditional lower bound under the Exponential Time Hypothesis (ETH). In the reduction to \twowriter\ systems, the number of threads depends both on $k$ (number of variables in the formula) and $m$ (the number of clauses). Hence this reduction cannot be used for the conditional lower bound. 
We will first describe the reduction to three writers and its implications, and then follow up with the second reduction. 

\paragraph*{Notation.} Let the variables used by $\varphi$ be $x_1, \dots, x_k$. The clauses are $C_1, \dots, C_m$. Each clause $C_j$ consists of three literals $(\ell_j^1, \ell_j^2, \ell_j^3)$, where each literal is either a variable $x_i$ or its negation $\neg x_i$. We define $S_{x_i} := \{j \mid C_j \text{ contains } x_i\}$ and $S_{\neg x_i} := \{ j \mid C_j \text{ contains } \neg x_i\}$ as the set of (indices of) clauses containing the literals $x_i$ and $\neg x_i$ respectively. 

\subsection{Reduction to \threewriter\ systems}
\label{sec:hardness-threewriter}
We consider execution graphs where for each variable $x$ there are at most three threads writing to $x$, that is, at most three threads can have statements of the form $\wt(x, d)$. We will call them \threewriter\ execution graphs. 
Our objective is to show a hardness result stronger than $\NP$-hardness, for the $\vm$ problem under the assumption stated below.

\noindent \textbf{Exponential Time Hypothesis (ETH)}. There is no $2^{o(k)}\cdot (k+m)^{\mathcal{O}(1)}$ algorithm for 3-CNF-SAT, where $k$ is the number of variables, $m$ is the number of clauses~\cite{ImpagliazzoPZ01,Lokshtanov:ETH:2011}.

We prove the following theorem in this section. 

\begin{theorem}\label{thm:eth-hardness}
    Under ETH, there is no $2^{o(k)} \cdot n^{\mathcal{O}(1)}$ algorithm for the $\vm$-problem with $\MemModel \in \{ \sramm, \ramm, \wramm\}$, $k$ being the number of threads, and $n$ the number of events. The result holds even while restricting the input instance to \threewriter\ execution graphs.
\end{theorem}

\paragraph*{Proof strategy.} Given a 3-CNF formula $\varphi$ with $k$ threads and $m$ clauses, we construct a partial execution graph $\expartial_\varphi = \tuple{\E, \po}$ that satisfies the following properties:
\begin{itemize}
\item $\expartial_\varphi$ is \threewriter,
\item the number of threads in $\expartial_\varphi$ is $\mathcal{O}(k)$,
\item the number of events in $\expartial_\varphi$ is $\mathcal{O}(k + m)$, 
\item $\expartial_\varphi \models \MemModel$ iff $\varphi$ is satisfiable, for $\MemModel \in \{\sramm, \ramm, \wramm\}$. 
\end{itemize} 

\begin{wrapfigure}[]{r}{0.48\textwidth}
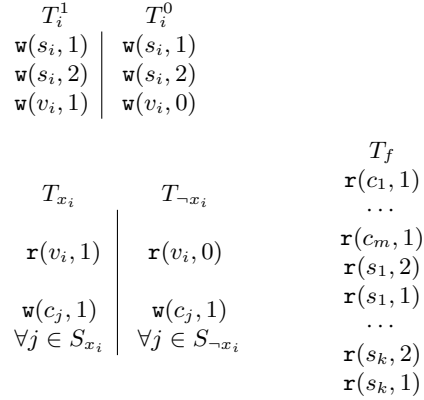
 
\medskip
\footnotesize
\begin{minipage}{0.1\textwidth}
\centering
\begin{tabular}{c| c}
    $T_i^1$ & $T_i^0$  \\
    $\wt(s_i, 1)$~~&~~$\wt(s_i, 1)$~~ \\
    $\wt(s_i, 2)$~~&~~$\wt(s_i, 2)$~~ \\
    $\wt(v_i, 1)$~~&~~$\wt(v_i, 0)$~~
\end{tabular}
\end{minipage}
\hfill
\medskip

\begin{minipage}{0.05\textwidth}
\centering
\begin{tabular}{c | c }
    $T_{x_i}$~ & ~$T_{\neg x_i}$ \\
        & \\
        $\rd(v_i, 1)$~&~$\rd(v_i, 0)$ \\
        & \\
        $\wt(c_j, 1)$ ~&~ $\wt(c_j, 1)$ \\
        $\forall j \in S_{x_i}$ ~&~  $\forall j \in S_{\neg x_i}$ 
\end{tabular}
\end{minipage}
\hfill
\medskip
\begin{minipage}{0.15\textwidth}
\centering
\begin{tabular}{c}
    $T_f$ \\
    $\rd(c_1, 1)$ \\
    $\cdots$ \\
    $\rd(c_m, 1)$ \\
    $\rd(s_1, 2)$ \\
    $\rd(s_1, 1)$ \\
    $\cdots$ \\
    $\rd(s_k, 2)$ \\
    $\rd(s_k, 1)$ 
\end{tabular}
\end{minipage}
 \caption{Threads in the reduction, for $i \in \{1, 2, \dots, k\}$ }
   \label{tab:reduction-threads-ra}
\end{wrapfigure}
Thus, if there was a $2^{o(k)} \cdot n^{\mathcal{O}(1)}$ algorithm for the $\vm$ problem, we could solve 3-CNF-SAT in $2^{o(k)}.(k+m)^{\mathcal{O}(1)}$ time: from $\varphi$, construct the $\vm$ instance in $\mathcal{O}(k + m)$ time, and then use the $2^{o(k)} \cdot (k+m)^{\mathcal{O}(1)}$ procedure for $\vm$, thereby contradicting ETH. To show the last item above, we make use of what is known in the literature as a \emph{range reduction}~\cite{ChakrabortyKMP24}: using Lemma~\ref{lem:sra-ra-wra-relationship}, it is enough to show that if $\varphi$ has a satisfying assignment, then $\expartial \models \sramm$, and if $\varphi$ does not have a satisfying assignment, then $\expartial \not \models \wramm$. 

\paragraph*{The reduction.} 
 For every variable $x_i$, we introduce four threads $T_i^1$, $T_i^0$, $T_{x_i}$ and $T_{\neg x_i}$, as shown in Figure~\ref{tab:reduction-threads-ra}. Apart from these, there is a special thread $T_f$. 
 A concrete example of the reduction is presented in Appendix~\ref{app:hardness-example}. 

Here is an intuitive description of the reduction. The statements $\wt(v_i, 1)$ and $\wt(v_i,0)$ in $T_i^1$ and $T_i^0$ encode variable $x_i$ being set to \emph{true} and \emph{false} respectively. Thread $T_{x_i}$ has a read statement $\rd(v_i,1)$ (reading \emph{true} for $x_i$) and sets all clauses containing $x_i$ to true, using the statement $\wt(c_j, 1)$ for all $j \in S_{x_i}$. Similarly, thread $T_{\neg x_i}$ reads $\rd(v_i, 0)$ (\emph{false} for $x_i$) and sets all clauses containing $\neg x_i$ to true. Finally the thread $T_{f}$ needs to make sure that all clauses are set to true, hence the read statements $\rd(c_1, 1), \dots, \rd(c_m, 1)$.  The main challenge is to encode a valid assignment, where every variable is assigned a unique truth value: to do this, we need to ensure that there are no two distinct reads $\rd(c_{j}, 1)$ and $\rd(c_\ell, 1)$ in $T_f$ which read from some $T_{x_i}$ and its negation $T_{\neg x_i}$.  This is where the rest of the statements, using $s_i$ become useful, as we elaborate in the proofs below.

\begin{lemma}\label{lem:wra-eth}
If  $\expartial_\varphi \models \wramm$, then $\varphi$ has a satisfying assignment. 
\end{lemma}
\begin{proof} 
    If $\expartial_\varphi \models \wramm$ there is an $\rf$ function which satisfies porf-acyclicity and weak-read-coherence. We need to prove that this $\rf$ encodes a valid satisfying assignment. This can be shown by proving that there are no two distinct reads $\rd(c_j, 1)$ and $\rd(c_\ell, 1)$ which read from $T_{x_i}$ and $T_{\neg x_i}$ for some variable $x_i$, since it gives the following assignment: $A(x_i) = true$ if some read event $\rd(c_j, 1)$ reads from the write in $T_{x_i}$; $A(x_i) = false$ otherwise. Moreover, the assignment $A$ is satisfying since every clause is made  true by this assignment.

    Let us now prove the property. Suppose $\rd(c_j, 1)$ and $\rd(c_\ell, 1)$ read from $T_{x_i}$ and $T_{\neg x_i}$ respectively. Since $\rd(c_j, 1)$ reads from $T_{x_i}$, there is an $\hb$ path from $T_i^1$ to $\rd(c_j, 1)$ due to the $\rf$ edges: $\wt(v_i, 1) \xra{\rf} \rd(v_i, 1)$ and $T_{x_i}: \wt(c_j, 1) \xra{\rf} \rd(c_j, 1)$. Similarly, there is an $\hb$ path from $T_i^0$ to $\rd(c_\ell, 1)$. Hence there is an $\hb$ path from both $T_i^1$ and $T_i^0$ to the event $\rd(s_i, 1)$ in $T_f$. Since the only matching write $\wt(s_i, 1)$ appears above $\wt(s_i, 2)$ in these threads, any assignment to $\rd(s_i, 1)$ violates weak-read-coherence, a contradiction to the hypothesis that $\rf$ satisfies weak-read-coherence.
\end{proof}

\begin{restatable}{lemma}{lemsatimpliessra}
    \label{lem:sat-implies-sra}
If $\varphi$ has a satisfying assignment, then $\expartial_\varphi \models \sramm$. 
\end{restatable}
\begin{proof}
Suppose $\varphi$ has a satisfying assignment $A$. We need to synthesize $\rf$ and $\mo$ such that the resulting execution graph satisfes porf-acyclicity, strong-write-coherence and strong-read-coherence. Since $A$ is a satisfying assignment, for every clause $C_j$ there is a literal $\ell_j$ that is made true. Let $\rf^{-1}(\rd(c_j, 1))$ be the write event $\wt(c_j, 1)$ in $T_{\ell_j}$. Assigning writes this way introduces an $\hb$ path from $T_i^{A(x_i)}$ to $T_f$ for all variables that are indeed being utilized to make the clauses true. Importantly, there is no $\hb$ path created from $T_{i}^{\neg A(x_i)}$ to the thread $T_f$. Therefore, the read event $\rd(s_i, 2)$ can read from the write in $T_i^{A(x_i)}$ and $\rd(s_i, 1)$ can read from the unused thread $T_i^{\neg A(x_i)}$. This gives an $\rf$ which satisfies weak-read-coherence. By design of the execution graph, the $\rf$ also satisfies porf-acyclicity. For the lemma, we require to prove strong-write-coherence and read-coherence. For this we need to also synthesize an $\mo$. 
Consider any  $\mo$ relation that respects the following orders:
    \begin{itemize}[leftmargin=0.1em]
    \item $T_{i}^{A(x_i)}: \wt(s_i, 1) \xra{\mo} T_i^{\neg A(x_i)}: \wt(s_i, 1) \xra{\mo}T_{i}^{\neg A(x_i)}: \wt(s_i, 2) \xra{\mo} T_i^{A(x_i)}: \wt(s_i, 2)$,
    \item For each clause $j$ there are three threads writing $\wt(c_j, 1)$.  For each clause $C_j$ choose a specific literal $\ell_j$ that is made true by the satisfying assignment $A$. The $\mo$ relation orders the other write statements $\wt(c_j, 1)$ before $T_{\ell_j}: \wt(c_j, 1)$. 
    \end{itemize}  
It can be shown that the constructed $\rf$ and $\mo$ satisfy strong-write-coherence and read-coherence (see Appendix~\ref{app:hardness} for the full proof).
\end{proof}

\subsection{Reduction to \twowriter\ systems}
\label{sec:hardness-twowriter}

The reduction in Section~\ref{sec:hardness-threewriter} involved \threewriter\ systems.
Mainly, each clause variable $c_j$ was written by three threads corresponding to its three literals. Here, our objective is to prove hardness even when there are atmost two threads writing to each variable. We will now reduce 3-CNF-SAT to the consistency testing problem for \twowriter~ systems. However, in this reduction, a formula $\varphi$ with $k$ threads and $m$ clauses yields an execution graph with $\mathcal{O}(k + m)$ theads. While this still proves $\NP$-hardness, it does not give the stronger hardness under ETH, parameterized by the number of threads. 

\begin{theorem}
\label{thm:two-writer-hardness}
    The $\vm$ problem is $\NP$-complete for $\MemModel \in \{\sramm, \ramm, \wramm\}$, even when restricted to \twowriter\ execution graphs.
\end{theorem}

The key challenge in proving the hardness is to reduce the number of threads writing to each clause variable $c_j$ from three (one for each literal in the clause) to at most two, while still encoding the satisfiability of the 3-CNF formula. This can be achieved by introducing one auxiliary variable $d_j$ per clause. Suppose $C_j = (\ell_j^1 \vee \ell_j^2 \vee \ell_j^3)$. In all the reductions of Section~\ref{sec:hardness-threewriter}, there are three threads $T_{\ell_j^1}, T_{\ell_j^2}$ and $T_{\ell_j^3}$. We make the following changes to constructions of $\expartial_\varphi$ in Section~\ref{sec:hardness-threewriter}.

\begin{itemize}[leftmargin=0.1em]
\item Replace $\wt(c_j, 1)$ in $T_{\ell_j^3}$ with $\wt(d_j, 1)$, i.e., in the third literal, the write is on $d_j$.
\item Add a new thread $T_{j}$ with two events as follows: $\rd(c_j, 1)~;~\wt(d_j, 1)$
\item Finally, in $T_f$, replace all $\rd(c_j, 1)$ with $\rd(d_j, 1)$
\end{itemize}

Notice that due to the introduction of the new threads $T_j$, the total number of threads also depends on the number of clauses $m$. Let the resulting partial execution graph be called $\ypartial_\varphi$. In Appendix~\ref{sec:app-np-hardness} we prove the correctness of the reduction.
\section{Conclusion}
\label{sec:conclusion}

Testing shared memory implementations is a key problem in concurrency, whose computational foundations have been laid by the work of~\cite{Gibbons} for sequential consistency. Ever since, the work has been applied to various other memory models and also used in stateless model-checking algorithms for concurrent programs. We have contributed to the study of this problem under the release-acquire semantics. We have presented a trichotomy of the landscape, with \onewriter\ solved in polynomial-time, \twowriter\ being $\NP$-complete, and \threewriter\ having a stronger lower bound under the Exponential-Time-Hypothesis (ETH). Our presentation works for $\ramm$ and its close variants $\sramm$ and $\wramm$. 

A milder variant of the problem studies the complexity when the reads-from relation $\rf$ is also given. Then the question is to synthesize an $\mo$ such that the axioms can be satisfied. This ``$\rf$-given'' problem can be solved in polynomial-time~\cite{Tunc2023}, under no restrictions on the execution graphs. When $\rf$ is not provided apriori, while the general problem is known to be $\NP$-hard~\cite{ChakrabortyKMP24}, no subclass or bounded variant with a polynomial-time algorithm is known in the literature (to the best of our knowledge). We have shown that the \onewriter\ subclass lends itself to a polynomial-time algorithm.

Our results can also be extended to the memory model called $\rlxmm$, whose consistency axioms are a relaxed version of write-coherence and read-coherence, where essentially $\hb$ is replaced with $\po$. This makes the consistency check easier for the \onewriter\ algorithm. For the hardness, we need to make small changes, but we still obtain the same trichotomy of results. Appendix~\ref{app:relaxed} provides the details. Another set of memory models that are closely related w.r.t the analysis techniques are the models $\ccmm$, $\cmmm$ and $\cvmm$ coming from distributed databases. It is known that $\ccmm$ coincides with $\wramm$ while $\cvmm$ coincides with $\sramm$~\cite{Lahav:2022}. Some care needs to be taken for $\cmmm$, but yet again, we get the same trichotomy of results. Appendix~\ref{app:cc-models} has the details. 

Adapting these complexity theoretic insights to improve algorithmics for the verification of concurrent programs under the release-acquire semantics, is one direction for future work. 

\bibliographystyle{splncs04}
\bibliography{m}

@article{Ahamad1995,
author = {Ahamad, Mustaque and Neiger, Gil and Burns, James E. and Kohli, Prince and Hutto, Phillip W.},
title = {Causal Memory: Definitions, Implementation, and Programming},
year = {1995},
issue_date = {March     1995},
publisher = {Springer-Verlag},
address = {Berlin, Heidelberg},
volume = {9},
number = {1},
issn = {0178-2770},
url = {https://doi.org/10.1007/BF01784241},
doi = {10.1007/BF01784241},
journal = {Distrib. Comput.},
month = {mar},
pages = {37-49},
numpages = {13}
}

@article{Lahav:2022,
author = {Lahav, Ori and Boker, Udi},
title = {What's Decidable About Causally Consistent Shared Memory?},
year = {2022},
volume = {44},
number = {2},
doi = {10.1145/3505273},
journal = {ACM Trans. Program. Lang. Syst.},
month = {apr},
articleno = {8},
numpages = {55},
}

@article{Tunc2023,
author = {Tun\c{c}, H\"{u}nkar Can and Abdulla, Parosh Aziz and Chakraborty, Soham and Krishna, Shankaranarayanan and Mathur, Umang and Pavlogiannis, Andreas},
title = {Optimal Reads-From Consistency Checking for C11-Style Memory Models},
year = {2023},
issue_date = {June 2023},
publisher = {Association for Computing Machinery},
address = {New York, NY, USA},
volume = {7},
number = {PLDI},
url = {https://doi.org/10.1145/3591251},
doi = {10.1145/3591251},
journal = {Proc. ACM Program. Lang.},
month = {jun},
articleno = {137},
numpages = {25}
}

@article{Burckhardt:2014,
	title        = {{Principles of Eventual Consistency}},
	author       = {Sebastian Burckhardt},
	year         = 2014,
	journal      = {Foundations and Trends in Programming Languages},
	volume       = 1,
	number       = {1-2},
	pages        = {1--150},
	doi          = {10.1561/2500000011}
}

@inbook{Luo:2021,
author = {Luo, Weiyu and Demsky, Brian},
title = {C11Tester: A Race Detector for C/C++ Atomics},
year = {2021},
publisher = {Association for Computing Machinery},
address = {New York, NY, USA},
doi = {10.1145/3445814.3446711},
booktitle={ASPLOS'21},
pages={630-646},
}

@inproceedings{Lahav:2017,
 author = {Lahav, Ori and Vafeiadis, Viktor and Kang, Jeehoon and Hur, Chung-Kil and Dreyer, Derek},
 title = {Repairing Sequential Consistency in {C/C++11}},
 booktitle = {PLDI 2017},
 year = {2017},
 pages = {618--632},
 doi = {10.1145/3062341.3062352},
 note ={Technical Appendix Available at \url{https://plv.mpi-sws.org/scfix/full.pdf}},
}

@inproceedings{Batty:2011,
 author = {Batty, Mark and Owens, Scott and Sarkar, Susmit and Sewell, Peter and Weber, Tjark},
 title = {Mathematizing {C++} concurrency},
 booktitle = {POPL'11},
 year = {2011},
 pages     = {55--66},
 publisher = {{ACM}},
 doi  = {10.1145/1926385.1926394},
}

@article{WilliamsW18,
  author       = {Virginia {Vassilevska Williams} and
                  R. Ryan Williams},
  title        = {Subcubic Equivalences Between Path, Matrix, and Triangle Problems},
  journal      = {J. {ACM}},
  volume       = {65},
  number       = {5},
  pages        = {27:1--27:38},
  year         = {2018}
}

@article{Gibbons,
  author       = {Phillip B. Gibbons and
                  Ephraim Korach},
  title        = {Testing Shared Memories},
  journal      = {{SIAM} J. Comput.},
  volume       = {26},
  number       = {4},
  pages        = {1208--1244},
  year         = {1997}
}

@inproceedings{Mathur20,
  author       = {Umang Mathur and
                  Andreas Pavlogiannis and
                  Mahesh Viswanathan},
  title        = {The Complexity of Dynamic Data Race Prediction},
  booktitle    = {{LICS}},
  pages        = {713--727},
  publisher    = {{ACM}},
  year         = {2020}
}

@inproceedings{rv-equivalence,
  author       = {Pratyush Agarwal and
                  Krishnendu Chatterjee and
                  Shreya Pathak and
                  Andreas Pavlogiannis and
                  Viktor Toman},
  title        = {Stateless Model Checking Under a Reads-Value-From Equivalence},
  booktitle    = {{CAV} {(1)}},
  series       = {Lecture Notes in Computer Science},
  volume       = {12759},
  pages        = {341--366},
  publisher    = {Springer},
  year         = {2021}
}

@article{ImpagliazzoPZ01,
  author       = {Russell Impagliazzo and
                  Ramamohan Paturi and
                  Francis Zane},
  title        = {Which Problems Have Strongly Exponential Complexity?},
  journal      = {J. Comput. Syst. Sci.},
  volume       = {63},
  number       = {4},
  pages        = {512--530},
  year         = {2001}
}

@inproceedings{fine-grained-complexity,
  author       = {Peter Chini and
                  Prakash Saivasan},
  editor       = {Nitin Saxena and
                  Sunil Simon},
  title        = {A Framework for Consistency Algorithms},
  booktitle    = {40th {IARCS} Annual Conference on Foundations of Software Technology
                  and Theoretical Computer Science, {FSTTCS} 2020, December 14-18, 2020,
                  {BITS} Pilani, {K} {K} Birla Goa Campus, Goa, India (Virtual Conference)},
  series       = {LIPIcs},
  volume       = {182},
  pages        = {42:1--42:17},
  publisher    = {Schloss Dagstuhl - Leibniz-Zentrum f{\"{u}}r Informatik},
  year         = {2020},
  url          = {https://doi.org/10.4230/LIPIcs.FSTTCS.2020.42},
  doi          = {10.4230/LIPICS.FSTTCS.2020.42},
  bibsource    = {dblp computer science bibliography, https://dblp.org}
}

@article{SC-Lamport78,
  author       = {Leslie Lamport},
  title        = {Time, Clocks, and the Ordering of Events in a Distributed System},
  journal      = {Commun. {ACM}},
  volume       = {21},
  number       = {7},
  pages        = {558--565},
  year         = {1978}
}

@article{ChakrabortyKMP24,
  author       = {Soham Chakraborty and
                  Shankara Narayanan Krishna and
                  Umang Mathur and
                  Andreas Pavlogiannis},
  title        = {How Hard Is Weak-Memory Testing?},
  journal      = {Proc. {ACM} Program. Lang.},
  volume       = {8},
  number       = {{POPL}},
  pages        = {1978--2009},
  year         = {2024}
}

@inproceedings{BouajjaniEGH17,
  author       = {Ahmed Bouajjani and
                  Constantin Enea and
                  Rachid Guerraoui and
                  Jad Hamza},
  title        = {On verifying causal consistency},
  booktitle    = {{POPL}},
  pages        = {626--638},
  publisher    = {{ACM}},
  year         = {2017}
}

@inproceedings{ChenLHCSWP09,
  author       = {Yunji Chen and
                  Yi Lv and
                  Weiwu Hu and
                  Tianshi Chen and
                  Haihua Shen and
                  Pengyu Wang and
                  Hong Pan},
  title        = {Fast complete memory consistency verification},
  booktitle    = {{HPCA}},
  pages        = {381--392},
  publisher    = {{IEEE} Computer Society},
  year         = {2009}
}

@inproceedings{ManovitH06,
  author       = {Chaiyasit Manovit and
                  Sudheendra Hangal},
  title        = {Completely verifying memory consistency of test program executions},
  booktitle    = {{HPCA}},
  pages        = {166--175},
  publisher    = {{IEEE} Computer Society},
  year         = {2006}
}

@article{Qadeer03,
  author       = {Shaz Qadeer},
  title        = {Verifying Sequential Consistency on Shared-Memory Multiprocessors
                  by Model Checking},
  journal      = {{IEEE} Trans. Parallel Distributed Syst.},
  volume       = {14},
  number       = {8},
  pages        = {730--741},
  year         = {2003}
}

@article{tso_rf_bui2021reads,
  title={The reads-from equivalence for the TSO and PSO memory models},
  author={Bui, Truc Lam and Chatterjee, Krishnendu and Gautam, Tushar and Pavlogiannis, Andreas and Toman, Viktor},
  journal={Proceedings of the ACM on Programming Languages},
  volume={5},
  number={OOPSLA},
  pages={1--30},
  year={2021},
  publisher={ACM New York, NY, USA}
}

@article{Lahav2021,
  author       = {Ori Lahav and
                  Udi Boker},
  title        = {What's Decidable About Causally Consistent Shared Memory?},
  journal      = {{ACM} Trans. Program. Lang. Syst.},
  volume       = {44},
  number       = {2},
  pages        = {8:1--8:55},
  year         = {2022},
  url          = {https://doi.org/10.1145/3505273},
  doi          = {10.1145/3505273},
  bibsource    = {dblp computer science bibliography, https://dblp.org}
}

@book{Jaja1992ParallelAlgorithms,
  author    = {Joseph J{\'a}J{\'a}},
  title     = {An Introduction to Parallel Algorithms},
  publisher = {Addison--Wesley},
  year      = {1992},
  isbn      = {978-0201548568}
}

@article{KLSV:popl18,
author = {Kokologiannakis, Michalis and Lahav, Ori and Sagonas, Konstantinos and Vafeiadis, Viktor},
title = {Effective stateless model checking for C/C++ concurrency},
year = {2017},
issue_date = {January 2018},
publisher = {Association for Computing Machinery},
address = {New York, NY, USA},
volume = {2},
number = {POPL},
url = {https://doi.org/10.1145/3158105},
doi = {10.1145/3158105},
journal = {Proc. ACM Program. Lang.},
month = dec,
articleno = {17},
numpages = {32}
}

@article{Kokologiannakis:POPL:2017,
author = {Kokologiannakis, Michalis and Lahav, Ori and Vafeiadis, Viktor},
title = {Kater: Automating Weak Memory Model Metatheory and Consistency Checking},
year = {2023},
issue_date = {January 2023},
publisher = {Association for Computing Machinery},
address = {New York, NY, USA},
volume = {7},
number = {POPL},
url = {https://doi.org/10.1145/3571212},
doi = {10.1145/3571212},
journal = {Proc. ACM Program. Lang.},
month = jan,
articleno = {19},
numpages = {29}
}

@article{Margalit:2021:POPL,
author = {Margalit, Roy and Lahav, Ori},
title = {Verifying observational robustness against a c11-style memory model},
year = {2021},
issue_date = {January 2021},
publisher = {Association for Computing Machinery},
address = {New York, NY, USA},
volume = {5},
number = {POPL},
url = {https://doi.org/10.1145/3434285},
doi = {10.1145/3434285},
journal = {Proc. ACM Program. Lang.},
month = jan,
articleno = {4},
numpages = {33}
}

@inproceedings{Norris:2013:OOPSLA,
author = {Norris, Brian and Demsky, Brian},
title = {CDSchecker: checking concurrent data structures written with C/C++ atomics},
year = {2013},
isbn = {9781450323741},
publisher = {Association for Computing Machinery},
address = {New York, NY, USA},
url = {https://doi.org/10.1145/2509136.2509514},
doi = {10.1145/2509136.2509514},
booktitle = {Proceedings of the 2013 ACM SIGPLAN International Conference on Object Oriented Programming Systems Languages \& Applications},
pages = {131–150},
numpages = {20},
location = {Indianapolis, Indiana, USA},
series = {OOPSLA '13}
}

@article{Lokshtanov:ETH:2011,
  author       = {Daniel Lokshtanov and
                  D{\'{a}}niel Marx and
                  Saket Saurabh},
  title        = {Lower bounds based on the Exponential Time Hypothesis},
  journal      = {Bull. {EATCS}},
  volume       = {105},
  pages        = {41--72},
  year         = {2011},
  url          = {http://eatcs.org/beatcs/index.php/beatcs/article/view/92},
  bibsource    = {dblp computer science bibliography, https://dblp.org}
}

\appendix

\section{Appendix for Section~\ref{sec:onewriter}}

\subsection{Proofs of Section~\ref{sec:onewriter-structural-properties}}

\lemonewriterreducestowra*

\begin{proof}
    Let us begin with $\MemModel = \ramm$. For $\ex$ to be consistent with $\ramm$, it has to satisfy the three axioms: porf-acyclicity, write-coherence and read-coherence~(see Table~\ref{table:ax:models-2}). In a \onewriter\ system, for $\mo$ edges, both its source and target are in the same thread. Therefore, to satisfy write-coherence, $\mo$ should agree with $\po$: in other words, if $w_1~\mo_x~w_2$, then $w_1~\po^+~w_2$. Hence, notice that to prove the lemma, it is sufficient to show that the read-coherence axiom and weak-read-coherence axiom coincide for \onewriter. A violating cycle for read-coherence is of the form
    \[r \xra{\rf} w_1 \xra{\mo_x} w_2 \xra{\hb} r\]
    As seen above the $\mo_x$ edge is simply $\po^+$, hence $w_1 \xra{\mo_x} w_2$ can be replaced with $w_1 \xra{\po^+} w_2$. A violating cycle for weak-read-coherence is of the form
    \[r \xra{\rf} w_1 \xra{\hb} w_2 \xra{\hb} r\]
    with $w_1$ and $w_2$ being writes on the same variable $x$. Since the writes are on the same variable, $w_1 \xra{\hb} w_2$ can be replaced with $w_1 \xra{\po^+} w_2$. This proves that the violating cycles for both the axioms coincide for \onewriter\ systems.

    Now, let $\MemModel = \sramm$. For strong-write-coherence to be true, we require $\mo$ to agree with $\po$. If $X \models \sramm$, then the axioms porf-acyclicity, strong-write-coherence and read-coherence are true. Since strong-write-coherence holds, $\mo$ needs to agree with $\po$. We have shown above that read-coherence coincides with weak-read-coherence, when $\mo$ agrees with $\po$. Hence $X \models \wramm$. Now, for the other direction. Suppose $\mo$ agrees with $\po$ and weak-read-coherence are true. From above, we deduce read-coherence holds. If strong-write-coherence is violated, then there will be a porf-cycle, a contradiction. Hence, the lemma follows for $\MemModel = \sramm$.
\end{proof}

\lemmaonewritertwoproperties*

\begin{proof}
For the first item, suppose $\rf_1$ induces a porf-cycle of the form:
\[ e_1 \xra{\rf_1} f_1 \xra{\po} e_2 \xra{\rf_1} f_2 \cdots e_m \xra{\rf_1} f_{m} \xra{\po} e_1\]
Notice that each $e_i$ is a write event and $f_i$ is a read event such that $e_i = \rf_1(f_i)$. Since $\rf_1 \sqsubset \rf_2$, we know $\rf_1(f_i) \xra{\po^*} \rf_2(f_i)$. Hence: $e_i \xra{\po^*} e'_i$ where $e'_i = \rf_2(f_i)$. This gives the following $\po \cup \rf_2$ cycle:
 \[ e_1 \xra{\po^*} e'_1 \xra{\rf_2} f_1 \xra{\po} e_2 \xra{\po^*} e'_2 \xra{\rf_2} f_2 \cdots e_m \xra{\po^*} e'_m \xra{\rf_2} f_{m} \xra{\po} e_1\]
 This proves the first item.

 For the second item, denote $\rf_0 = \min(\rf_1, \rf_2)$. We will show that if $\rf_0$ violates weak-read-coherence, then either $\rf_1$ or $\rf_2$ violates it. Suppose $\rf_0$ induces a violation of weak-read-coherence. There is a cycle of the form:
 \[ r \xra{\rf_0^{-1}} w \xra{\po^+} w' \xra{\hb_0} r \]
 where both $w$ and $w'$ are on the same variable, and hence the edge between them is $\po^+$. By definition, either $\rf_1$ or $\rf_2$ associates $r$ to $w$. Let us say $\rf_1(r) = w$. If we show $w' \hb_0 r$ then we are done. Since $w'~\hb_0~r$, there is a path with $\po$ and $\rf_0$ edges from $w'$ to $r$:
 \[w' \xra{\po^*} e_1 \xra{\rf_0} f_1 \xra{\po^*} e_2 \xra{\rf_0} f_2 \cdots e_m \xra{\rf_0} f_m \xra{\po^*} r \]
 We will now form a $\po \cup \rf_1$ path from $w'$ to $r$. For each $e_i \xra{\rf_0} f_i$, either $e_i \xra{\rf_1} f_i$ too, in which case we retain that edge. If not, by construction of $\rf_0$, we know that $e_i \xra{\rf_2} f_i$ and $e_i \po^* \rf_1(f_i)$. Replace $e_i \xra{\rf_0} f_i$ by the sequene $e_i \po^* \rf_1(f_i) \xra{\rf_1} f_i$ as in the proof of the first item, to deduce an $\hb_1$ path from $w'$ to $r$. This show that there is a violating cycle for weak-read-coherence induced by $\rf_1$ too, thereby proving the second item. 
\end{proof}

\subsection{Algorithm}
\label{app:algorithm}

\begin{algorithm}[h!]
\caption{$\wramm$-consistency checking in \onewriter\ systems}
\label{algo:one-writer-wra}
\Fn{$\CheckC(\expartial)$}{
    $\rf \gets \Initialize(\expartial)$ \\
    \While{$\rf$ violates weak-read-coherence}{\label{algo-while-main-onewriter}
        $\rf \gets \update(\expartial, \rf)$  \label{algo-update-call}
    }
    \Comment{Now $\rf = \rf_{\min}$}
    \If{$\rf$ satisfies porf-acyclicity}{ \label{algo-check-porf}
        \Return $\expartial$ is $\wramm$-consistent \label{algo-declare-consistent}
    }
    \Else {
        \Return $\expartial$ is not $\wramm$-consistent  \label{algo-porf-cycle}
    }
}

\Fn{$\Initialize(\expartial)$}{
\For{each read $r = \rd(x, v)$}{
$t \gets$  thread containing $r$ \\
$t_x \gets$ writer thread for $x$ \\
\If{$t \neq t_x$}{
    $\rf^{-1}_0(r) \gets$ $\po$-smallest write $w$ of the form $\wt(x,v)$ in $t_x$ \\
    \If{there is no such write}{
        \Return $\expartial$ is not $\wramm$-consistent \label{init-not-consistent1}
    }
}
\If{$t = t_x$}{
    $\rf^{-1}_0(r) \gets$ $\po$-smallest write $w$ of the form $\wt(x,v)$ in $t_x$ s.t. $w~\po^+r$ \\
    \If{there is no such write}{
        \Return $\expartial$ is not $\wramm$-consistent \label{init-not-consistent2}
    }
}
}
\Return $\rf_0$
}

\Fn{$\update(\expartial, \rf_i)$}{
    Pick a cycle violating weak-read-coherence: $r_i \xra{\rf_i^{-1}} w_i \xra{\hb_i} w'_i \xra{\hb_i} r_i$ \label{algo-violatecycle} \\
    Let $r_i = \rd(x, v)$, $w_i = \wt(x, v)$, $w'_i = \wt(x, v')$ \\
    $t \gets$ thread of $r_i$ \\
    $t_x \gets$ writer thread of variable $x$ \\
    \If{$t \neq t_x$}{
    $\rf^{-1}_{i+1}(r_i) \gets$ $\po$-smallest write $w''_i$ of the form $\wt(x,v)$ in $t_x$ s.t. $w'_i ~\po* w''_i$ \label{algo-update1} \\
    \If{there is no such write}{
        \Return $\expartial$ is not $\wramm$-consistent  \label{algo-update-notconsistent1}
    }
}
\If{$t = t_x$}{
    $\rf^{-1}_{i+1}(r_i) \gets$ $\po$-smallest write $w''_i$ of the form $\wt(x,v)$ in $t_x$ s.t. $w'_i ~\po^* w''_i ~\po^+r_i$ \label{algo-update2}\\
    \If{there is no such write}{
        \Return $\expartial$ is not $\wramm$-consistent \label{algo-update-notconsistent2}
    }
}
$\rf^{-1}_{i+1}(r') \gets \rf_i(r')$ for all $r' \neq r_i$ \\
\Return $\rf_{i+1}$
}
\end{algorithm}

\thmwracorrectness*
\begin{proof}
Let $\rf_0$ be the initial $\rf$ function assigned by Line 2 of the algorithm. Notice that if $\Initialize$ executes either Line~\ref{init-not-consistent1} or \ref{init-not-consistent2}, then there is a read for which no write can be assigned, and hence there is no $\rf$ function at all, and there is nothing more to prove. Let us assume $\Initialize$ returns an $\rf$ function $\rf_0$. 

Let $\rf_i$ be the $\rf$ function assigned by the $i^{th}$ call to the $\update$ function  in Line~\ref{algo-update-call}. Let us assume that in the call to $\update(\expartial, \rf_i)$, a read event $r_i$'s write changes from $w_i$ to $w''_i$, to result in $\rf_{i+1}$. 
By construction, $\rf_{i+1}$ differs from $\rf_i$ only at $r_i$. Furthermore, the $\update$ function ensures that  $w_i ~\po^+ w''_i$ (Lines~\ref{algo-update1} and \ref{algo-update2}). 

\emph{The algorithm terminates.} At each call to the $\update$ procedure, a read event $r_i$ is picked and is associated with a new write event $w_i''$ which is $\po$-later than the current write event $w_i$. Therefore, there cannot be more than $n$ calls to the  $\update$ function. 

\emph{The algorithm is correct.} Suppose the algorithm stops after $m$ calls to $\update$ function in Line 4. We will now inductively prove that $\rf_i \sqsubseteq \rf^\dagger$ for every  $\rf$ function $\rf^\dagger$ that satisfies weak-read-coherence. The property is trivially true for $\rf_0$. Assume the property to be true for $\rf_{i}$, that is $\rf_i \sqsubseteq \rf^\dagger$. We will show that the call to the update function preserves the property. Since $\update$ has been called, there is a cycle violating weak-read-coherence: $r_i \xra{\rf_i} w_i \xra{\hb_i} w'_i \xra{\hb_i} r_i$. 

Let $\hb^\dagger$ be the happens-before relation induced by $\rf^{\dagger}$. Since $\rf_i \sqsubseteq \rf^\dagger$, by an argument similar to the proof of the second item of Lemma~\ref{lem:onewriter-axioms-two-properties}, we deduce $w'_i \xra{\hb^{\dagger}} r_i$. As $\rf^{\dagger}$ satisfies weak-read-coherence, this would mean $w'_i ~\po^*~\rf^{\dagger}(r_i)$: since $w'_i$ happens before $r_i$, the event $\rf^\dagger(r_i)$ that writes to $r_i$ should be $\po$-later to $w'_i$. In Lines~\ref{algo-update1} and \ref{algo-update2} of the algorithm, we pick the earliest matching write event $w''_i$ that appears $\po$-after $w'_i$. Therefore, this would imply $w''_i~\po^* \rf^{\dagger}(r_i)$. For the other events $r'$, there is no change in the $\update$ function. Hence, we deduce $\rf_{i+1} \sqsubseteq \rf^{\dagger}$ for every $\rf$ function $\rf^{\dagger}$ that satisfies weak-read-coherence.

When the Algorithm stops at Lines~\ref{algo-update-notconsistent1} or \ref{algo-update-notconsistent2}, it has not found a matching write that is $\po$-later than $w'_i$. From the above invariant that $\rf_i \sqsubseteq \rf^{\dagger}$, this means there is no consistent $\rf^\dagger$ for this partial execution graph $\expartial$. The algorithm correctly returns the verdict of inconsistency in Lines \ref{algo-update-notconsistent1} or \ref{algo-update-notconsistent2}. When the algorithm reaches \ref{algo-check-porf}, an $\rf$ function $\rf_m$ satisfying weak-read-coherence is found. By the above invariant, $\rf_m = \rf_{\min}$, the smallest $\rf$ that satisfies weak-read-coherence. Then, we test for porf-acyclicity. From the first item of Lemma~\ref{lem:onewriter-axioms-two-properties}, if it does not satisfy porf-acyclicity, no bigger $\rf$ function satisfies it, and we can conclude that both weak-read-coherence and porf-acyclicity cannot be simultaneously satisfied. Hence the algorithm correctly returns that the partial execution graph is no $\wramm$-consistent in Line~\ref{algo-porf-cycle}.

\noindent \textit{Complexity.} As discussed before, there are atmost $n$ calls to $\update$. Before each call, weak-read-coherence needs to be checked. This can be posed as a graph problem: there are $n$ nodes corresponding to the events, and there are three kinds of edges $\po$, $\rf$ and $\rf^{-1}$; does there exist a cycle of the form $r \xra{\rf^{-1}} w_x \xra{\po^*} w'_x \xra{(\po \cup \rf)^*} r$?  For each read event $r$ and its corresponding write event $w$, we need to check if there exist a $w'$ such that (1) there is a path of $\po$ edges from $w$ to $w'$, and (2) there is a path from $w'$ to $r$ over a combination of $\po$ and $\rf$ edges. For each read, this can be achieved by in $\mathcal{O}(n)$ time by a graph search. Hence, overall, across all possible reads, this algorithm takes $\mathcal{O}(n^2)$ time. Now, once inside the $\update$ procedure, finding the next write takes $\mathcal{O}(n)$ time. Therefore, each call to update costs $\mathcal{O}(n^2)$ and there are $n$ calls. This gives a complexity of $\mathcal{O}(n^3)$. Furthermore, the $\Initialize$ procedure takes at most $\mathcal{O}(n^2)$: for each read, there is atmost $\mathcal{O}(n)$ time spent on finding the matching write. The final check for $\po \cup \rf$ cycle can be done in $\mathcal{)}(n)$. Hence we get an overall complexity of $\mathcal{O}(n^3)$. 
\end{proof}

\subsection{Super linear lower bound for checking weak-read-coherence}
\label{app:super-linear-lower-bound}
This section is devoted to the proof of Proposition~\ref{prop:wra-onewriter-lowerbound}. 

\noindent{\bf{Reduction}}. Given an undirected graph $G = (V_G, E_G)$ with $n$ vertices, we construct a partial execution graph $\expartial = \tup{\E, \po}$ 
with $|\E|=\mathcal{O}(V_G+E_G)$ such that $\expartial$ satisfies weak-read-coherence iff $G$ is triangle-free.
For simplicity, we let $V_G=\{1,\dots, n\}$. As the graph is undirected, we assume that if $(u,v) \in E$ then $(v, u) \in E$.

\smallskip 
\noindent \emph{Events, threads  and memory locations}. For each  vertex $v \in V_G$, we have threads $t_{a_v}, t_{b_v}, t_{c_v}$ and $t_{d_v}$. 
For each edge $(u,v) \in E_G$, we have locations $a_{uv}, b_{uv}, c_{uv}$. We also have locations $a_v, b_v, c_v$ for $v \in V_G$. 
The events are $\wt(a_v,0), \wt(a_v,1), \rd(a_v, 0)$, as well as $\{\wt(x_{uv},0), \rd(x_{uv},0) \mid x \in \{a,b,c\}\}$. 
We also have events $\wt(b_v),\wt(c_v)$ (the values do not matter).  
This gives us $4|V_G|$ threads, $3(|E_G|+|V_G|)$ locations and $5|V_G|+2|E_G|$ events. 

We now describe the threads. 
Let $v \in V_G$ and $x \in \{a,b, c\}$. Let $N_v=\{\alpha \mid (v, \alpha) \in E_G\}$.
Let $W(x_{vw},0), R(x_{vw},0)$ be  macros representing respectively, 
the sequence of events $\wt(x_{v\alpha_1},0) \dots \wt(x_{v\alpha_m},0)$ 
and $\rd(x_{v\alpha_1},0) \dots \rd(x_{v\alpha_m},0)$ such that 
$\alpha_1, \dots, \alpha_n \in N_v$. 

\begin{align*}
& t_{a_v}: \wt(a_v,0) \wt(a_v,1) W(a_{vw},0) \\
& t_{b_v}: R(a_{wv},0) ~w(b_v)~ W(b_{vw},0) \\
& t_{c_v}: R(b_{wv},0)~w(c_v)~W(c_{vw},0) \\
& t_{d_v}: R(c_{wv},0)\rd(a_v,0)
\end{align*}

The construction is such that each location is written into by a unique thread. Further, each location is accessed only by at most two threads. Each location (barring $a_v$) is written into exactly once by its writer thread and the $\rf$ is trivially determined. Thus, the constructed $\expartial$ qualifies as a single writer. 
The construction has the following properties. 
\begin{enumerate}
    \item $(x, y) \in E_G$ iff $\wt(a_{x},0)~\po~\wt(a_{x},1)~\hb~\wt(b_{y})$
    \item $(x, y) \in E_G$ iff $\wt(b_{x})~\hb~\wt(c_{y})$
    \item $(x, y) \in E_G$ iff $\wt(c_{x})~\hb~\rd(a_{y},0)$
\end{enumerate}

Thus, a path $x-y-z$ of length 2 results in $\wt(a_x,1)~\po~\wt(a_x,0)~\hb~\wt(b_y)~\hb~\wt(c_z)$, and  a
path $x-y-z-w$ of length 3 results in 
\[\wt(a_x,1)~\po~\wt(a_x,0)~\hb~\wt(b_y)~\hb~\wt(c_z)~\hb~\rd(a_w,0) \]

The key invariant behind the reduction is that for any $v \in V_G$, $\wt(a_v,1)~ \hb ~\rd(a_v,0)$ iff $G$ contains a triangle. The former 
violates weak-read-coherence. Figure \ref{fig:super-linear-lower-bound} shows the reduction on a graph with 3 vertices $\{1,2,3\}$, having edges between every distinct pair of vertices.

\begin{figure*}[!htbp]
\newcommand{\xtstep}{0.3}
\newcommand{\ytstep}{0.7}
\newcommand{\ybias}{-0.3 }
\newcommand{\xstep}{1.4}
\newcommand{\ystep}{-0.6}
\newcommand{\xtscale}{0.8}
\def\crossoutopacity{0.3}
\def \numevents{5}
\def\scale{0.73}
\centering
\scalebox{\scale}{
\begin{tikzpicture}[thick,font=\footnotesize,
pre/.style={<-,shorten >= 2pt, shorten <=2pt, very thick},
post/.style={->,shorten >= 3pt, shorten <=3pt,   thick},
seqtrace/.style={line width=1},
und/.style={very thick, draw=gray},
virt/.style={circle,draw=black!50,fill=black!20, opacity=0}]
\tikzstyle{event} = [rectangle, fill=white, inner sep=0, minimum height=4.6mm, minimum width=14mm, rounded corners]

\node[] (S11) at (-5.5*\xstep,.1) {\normalsize $t_{a_1}$};
\node[] (S12) at (-5.5*\xstep,\numevents * \ystep) {};
\node[] (S21) at (-4.5*\xstep,.1) {\normalsize $t_{a_2}$};
\node[] (S22) at (-4.5*\xstep,\numevents * \ystep) {};
\node[] (S31) at (-3.5*\xstep,.1) {\normalsize $t_{a_3}$};
\node[] (S32) at (-3.5*\xstep,\numevents * \ystep) {};
\node[] (S41) at (-2.5*\xstep,.1) {\normalsize $t_{b_1}$};
\node[] (S42) at (-2.5*\xstep,\numevents * \ystep) {};
\node[] (S51) at (-1.5*\xstep,.1) {\normalsize $t_{b_2}$};
\node[] (S52) at (-1.5*\xstep,\numevents * \ystep) {};
\node[] (S61) at (-.5*\xstep,.1) {\normalsize $t_{b_3}$};
\node[] (S62) at (-.5*\xstep,\numevents * \ystep) {};
\node[] (S71) at (.5*\xstep,.1) {\normalsize $t_{c_1}$};
\node[] (S72) at (.5*\xstep,\numevents * \ystep) {};
\node[] (S81) at (1.5*\xstep,.1) {\normalsize $t_{c_2}$};
\node[] (S82) at (1.5*\xstep,\numevents * \ystep) {};
\node[] (S91) at (2.5*\xstep,.1) {\normalsize $t_{c_3}$};
\node[] (S92) at (2.5*\xstep,\numevents * \ystep) {};
\node[] (S101) at (3.5*\xstep,.1) {\normalsize $t_{d_1}$};
\node[] (S102) at (3.5*\xstep,\numevents * \ystep) {};
\node[] (S111) at (4.5*\xstep,.1) {\normalsize $t_{d_2}$};
\node[] (S112) at (4.5*\xstep,\numevents * \ystep) {};
\node[] (S121) at (5.5*\xstep,.1) {\normalsize $t_{d_3}$};
\node[] (S122) at (5.5*\xstep,\numevents * \ystep) {};

\draw[seqtrace] (S11) to (S12);
\draw[seqtrace] (S21) to (S22);
\draw[seqtrace] (S31) to (S32);
\draw[seqtrace] (S41) to (S42);
\draw[seqtrace] (S51) to (S52);
\draw[seqtrace] (S61) to (S62);
\draw[seqtrace] (S71) to (S72);
\draw[seqtrace] (S81) to (S82);
\draw[seqtrace] (S91) to (S92);
\draw[seqtrace] (S101) to (S102);
\draw[seqtrace] (S111) to (S112);
\draw[seqtrace] (S121) to (S122);

\node[event] (11) at (-5.5*\xstep, 1*\ystep + 0*\ybias) {$\wt(a_1,0)$};
\node[event] (12) at (-5.5*\xstep, 2*\ystep + 0*\ybias) {$\wt(a_1,1)$};
\node[event] (13) at (-5.5*\xstep, 3*\ystep + 0*\ybias) {$\wt(a_{12},0)$};
\node[event] (14) at (-5.5*\xstep, 4*\ystep + 0*\ybias) {$\wt(a_{13},0)$};

\node[event] (11) at (-4.5*\xstep, 1*\ystep + 0*\ybias) {$\wt(a_2,0)$};
\node[event] (12) at (-4.5*\xstep, 2*\ystep + 0*\ybias) {$\wt(a_2,1)$};
\node[event] (13) at (-4.5*\xstep, 3*\ystep + 0*\ybias) {$\wt(a_{21},0)$};
\node[event] (14) at (-4.5*\xstep, 4*\ystep + 0*\ybias) {$\wt(a_{23},0)$};

\node[event] (11) at (-3.5*\xstep, 1*\ystep + 0*\ybias) {$\wt(a_3,0)$};
\node[event] (12) at (-3.5*\xstep, 2*\ystep + 0*\ybias) {$\wt(a_3,1)$};
\node[event] (13) at (-3.5*\xstep, 3*\ystep + 0*\ybias) {$\wt(a_{31},0)$};
\node[event] (14) at (-3.5*\xstep, 4*\ystep + 0*\ybias) {$\wt(a_{32},0)$};

\node[event] (11) at (-2.5*\xstep, 1*\ystep + 0*\ybias) {$\rd(a_{31},0)$};
\node[event] (12) at (-2.5*\xstep, 2*\ystep + 0*\ybias) {$\rd(a_{21},0)$};
\node[event] (12) at (-2.5*\xstep, 3*\ystep + 0*\ybias) {$\wt(b_1)$};
\node[event] (13) at (-2.5*\xstep, 4*\ystep + 0*\ybias) {$\wt(b_{12},0)$};
\node[event] (14) at (-2.5*\xstep, 5*\ystep + 0*\ybias) {$\wt(b_{13},0)$};

\node[event] (11) at (-1.5*\xstep, 1*\ystep + 0*\ybias) {$\rd(a_{32},0)$};
\node[event] (12) at (-1.5*\xstep, 2*\ystep + 0*\ybias) {$\rd(a_{12},0)$};
\node[event] (12) at (-1.5*\xstep, 3*\ystep + 0*\ybias) {$\wt(b_2)$};
\node[event] (13) at (-1.5*\xstep, 4*\ystep + 0*\ybias) {$\wt(b_{21},0)$};
\node[event] (14) at (-1.5*\xstep, 5*\ystep + 0*\ybias) {$\wt(b_{23},0)$};

\node[event] (11) at (-.5*\xstep, 1*\ystep + 0*\ybias) {$\rd(a_{23},0)$};
\node[event] (12) at (-.5*\xstep, 2*\ystep + 0*\ybias) {$\rd(a_{13},0)$};
\node[event] (12) at (-.5*\xstep, 3*\ystep + 0*\ybias) {$\wt(b_3)$};
\node[event] (13) at (-.5*\xstep, 4*\ystep + 0*\ybias) {$\wt(b_{31},0)$};
\node[event] (14) at (-.5*\xstep, 5*\ystep + 0*\ybias) {$\wt(b_{32},0)$};

\node[event] (11) at (.5*\xstep, 1*\ystep + 0*\ybias) {$\rd(b_{31},0)$};
\node[event] (12) at (.5*\xstep, 2*\ystep + 0*\ybias) {$\rd(b_{21},0)$};
\node[event] (12) at (.5*\xstep, 3*\ystep + 0*\ybias) {$\wt(c_1)$};
\node[event] (13) at (.5*\xstep, 4*\ystep + 0*\ybias) {$\wt(c_{12},0)$};
\node[event] (14) at (.5*\xstep, 5*\ystep + 0*\ybias) {$\wt(c_{13},0)$};

\node[event] (11) at (1.5*\xstep, 1*\ystep + 0*\ybias) {$\rd(b_{32},0)$};
\node[event] (12) at (1.5*\xstep, 2*\ystep + 0*\ybias) {$\rd(b_{12},0)$};
\node[event] (12) at (1.5*\xstep, 3*\ystep + 0*\ybias) {$\wt(c_2)$};
\node[event] (13) at (1.5*\xstep, 4*\ystep + 0*\ybias) {$\wt(c_{21},0)$};
\node[event] (14) at (1.5*\xstep, 5*\ystep + 0*\ybias) {$\wt(c_{23},0)$};

\node[event] (11) at (2.5*\xstep, 1*\ystep + 0*\ybias) {$\rd(b_{23},0)$};
\node[event] (12) at (2.5*\xstep, 2*\ystep + 0*\ybias) {$\rd(b_{13},0)$};
\node[event] (12) at (2.5*\xstep, 3*\ystep + 0*\ybias) {$\wt(c_3)$};
\node[event] (13) at (2.5*\xstep, 4*\ystep + 0*\ybias) {$\wt(c_{31},0)$};
\node[event] (14) at (2.5*\xstep, 5*\ystep + 0*\ybias) {$\wt(c_{32},0)$};

\node[event] (11) at (3.5*\xstep, 1*\ystep + 0*\ybias) {$\rd(c_{31},0)$};
\node[event] (12) at (3.5*\xstep, 2.5*\ystep + 0*\ybias) {$\rd(c_{21},0)$};
\node[event] (13) at (3.5*\xstep, 4*\ystep + 0*\ybias) {$\rd(a_{1},0)$};

\node[event] (11) at (4.5*\xstep, 1*\ystep + 0*\ybias) {$\rd(c_{32},0)$};
\node[event] (12) at (4.5*\xstep, 2.5*\ystep + 0*\ybias) {$\rd(c_{12},0)$};
\node[event] (13) at (4.5*\xstep, 4*\ystep + 0*\ybias) {$\rd(a_{2},0)$};

\node[event] (11) at (5.5*\xstep, 1*\ystep + 0*\ybias) {$\rd(c_{23},0)$};
\node[event] (12) at (5.5*\xstep, 2.5*\ystep + 0*\ybias) {$\rd(c_{13},0)$};
\node[event] (13) at (5.5*\xstep, 4*\ystep + 0*\ybias) {$\rd(a_{3},0)$};

\end{tikzpicture}
}
\caption{
A partial execution $\expartial$ for a graph $G=(\{1,2,3\}, \{(1,2), (2,1),(1,3), (3,1), (2,3), (3,2)\})$. 
We have $(\wt(a_1,0), \rd(a_1,0))\in \hb$  violating weak read coherence in $\expartial$.
\label{fig:super-linear-lower-bound}
}
\end{figure*}
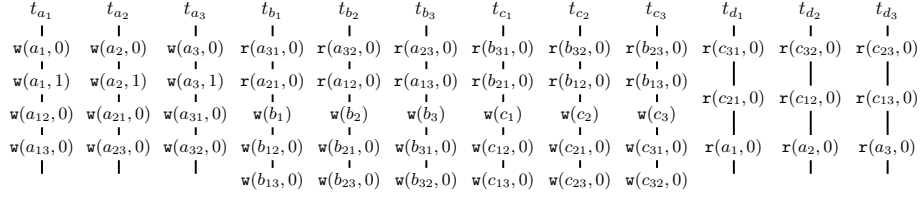

\noindent{\bf{Correctness and Time Complexity}}. Assume that $G$ has a triangle $(x, y, z)$. Wlg, let $x < y < z$. Then $\wt(a_x,1) ~\hb~ \rd(a_x,0)$ as follows: 
$\wt(a_x,1) ~\po ~\wt(a_{xy},0)$ $~\rf~ \rd(a_{xy},0)~\po~\wt(b_{yz},0)~\rf ~$ $\rd(b_{yz},0)~\po~\wt(c_{zx},0)~\rf~\rd(c_{zx},0)~\po~\rd(a_x,0)$. 
Together with $\wt(a_x,0)~\po~\wt(a_x,1)$, we obtain a weak-read-coherence violation. 
 Triangle-freeness in the input graph $G$ ensures satisfying weak-read-coherence. In the construction, the only way to violate weak-read-coherence 
is to have $\wt(a_v,0)~\po~\wt(a_v,1)~\hb ~\rd(a_v,0)$ for some $v \in V_G$. This is achieved by:
\[ \wt(a_v,1)~\po~\wt(a_{vu},0)~\hb~\wt(b_{uw},0)~\hb~\wt(c_{wv},0)~\rf~\rd(c_{wv},0)~\po~\rd(a_v,0)\] 

This points to a triangle $v u w$ in $G$.   
The total time to construct $\expartial$ is $O(V_G+E_G)$.

\section{Appendix for Section~\ref{sec:hardness}}
\label{app:hardness}

\subsection{An example to illustrate the reduction}
\label{app:hardness-example}
Table~\ref{tab:example-finegrainedwra} gives an instance of the fine-grained hardness reduction for \threewriter~. 
A satisfying assignment for $\varphi$ is $x_1=1, x_2=0, x_3=0$. The $\rf$ synthesis would be as follows:
\begin{itemize}
    \item $\tuple{T_{x_1}:\wt(c_1,1)}~\rf~\tuple{T_f:\rd(c_1,1)}$, $\tuple{T_{x_1}:\wt(c_2,1)}~\rf~\tuple{T_f:\rd(c_2,1)}$
    \item $\tuple{T_1^1:\wt(s_1,2)}~\rf~\tuple{T_f:\rd(s_1,2)}$, $\tuple{T_1^0:\wt(s_1,1)}~\rf~\tuple{T_f:\rd(s_1,1)}$
    \item $\tuple{T_2^0:\wt(s_2,2)}~\rf~\tuple{T_f:\rd(s_2,2)}$, $\tuple{T_2^1:\wt(s_2,1)}~\rf~\tuple{T_f:\rd(s_2,1)}$
    \item $\tuple{T_3^0:\wt(s_3,2)}~\rf~\tuple{T_f:\rd(s_3,2)}$, $\tuple{T_3^1:\wt(s_3,1)}~\rf~\tuple{T_f:\rd(s_3,1)}$
\end{itemize}

All other $\rf$-edges are trivial. Note that the resulting execution graph satisfies all the weak memory axioms of $\sramm$ (Refer Table~\ref{table:ax:models-2}).
\begin{table}[h]
    
\noindent
\begin{minipage}{0.15\textwidth}
\centering
\begin{tabular}{c | c}
    $T_1^1$ & $T_1^0$ \\
    $\wt(s_1,1)$~& ~$\wt(s_1,1)$ \\ 
    $\wt(s_1,2)$~& ~$\wt(s_1,2)$ \\ 
    $\wt(v_1, 1)$~& ~$\wt(v_1, 0)$ \\

\end{tabular}
\vspace{.1cm}

\begin{tabular}{c | c}
    $T_2^1$ & $T_2^0$ \\
    $\wt(s_2,1)$~& ~$\wt(s_2,1)$ \\ 
    $\wt(s_2,2)$~& ~$\wt(s_2,2)$ \\ 
    $\wt(v_2, 1)$~& ~$\wt(v_2, 0)$ \\
\end{tabular}
\vspace{.1cm}

\begin{tabular}{c | c}
    $T_3^1$ & $T_3^0$ \\
    $\wt(s_3,1)$~& ~$\wt(s_3,1)$ \\ 
    $\wt(s_3,2)$~& ~$\wt(s_3,2)$ \\ 
    $\wt(v_3, 1)$~& ~$\wt(v_3, 0)$ \\
   
\end{tabular}

\end{minipage}
\hfill
\begin{minipage}{0.15\textwidth}
\centering
\centering
\begin{tabular}{c | c}
    $T_{x_1}$ & $T_{\neg x_1}$ \\
    $\rd(v_1,1)$ ~& ~$\rd(v_1,0)$ \\ 
    $\wt(c_1,1)$~ &  \\ 
    $\wt(c_2, 1)$~&  \\

\end{tabular}
\vspace{.1cm}

\begin{tabular}{c | c}
   $T_{x_2}$ & $T_{\neg x_2}$ \\
    $\rd(v_2,1)$ ~& ~$\rd(v_2,0)$ \\ 
    $\wt(c_1,1)$~ & ~$\wt(c_2,1)$ \\ 
    
\end{tabular}
\vspace{.1cm}

\begin{tabular}{c | c}
    $T_{x_3}$ & $T_{\neg x_3}$ \\
    $\rd(v_3,1)$ ~& ~$\rd(v_3,0)$ \\ 
    $\wt(c_1,1)$~ & ~$\wt(c_2,1)$ \\

\end{tabular}
\end{minipage}
\hfill
\begin{minipage}{0.15\textwidth}
\centering
\begin{tabular}{c}
    $T_f$ \\
    $\rd(c_1, 1)$ \\
    $\rd(c_2, 1)$ \\
    $\rd(s_1, 2)$ \\
    $\rd(s_1, 1)$ \\
    $\rd(s_2, 2)$ \\
    $\rd(s_2, 1)$ \\
    $\rd(s_3, 2)$ \\
    $\rd(s_3, 1)$ \\

\end{tabular}
\end{minipage}
\vspace{2em}
\caption{$\varphi=(x_1\lor x_2\lor x_3)\land (x_1 \lor \neg x_2\lor \neg x_3)$}
    \label{tab:example-finegrainedwra}
\end{table}

\subsection{Missing proofs of Section~\ref{sec:hardness-threewriter}}
\lemsatimpliessra*
\begin{proof}
Suppose $\varphi$ has a satisfying assignment $A$. We need to synthesize $\rf$ and $\mo$ such that the resulting execution graph satisfes porf-acyclicity, strong-write-coherence and strong-read-coherence. Since $A$ is a satisfying assignment, for every clause $C_j$ there is a literal $\ell_j$ that is made true. Let $\rf(\rd(c_j, 1))$ be the write event $\wt(c_j, 1)$ in $T_{\ell_j}$. Assigning writes this way introduces an $\hb$ path from $T_i^{A(x_i)}$ for all variables that are indeed being utilized to make the clauses true. Importantly, there is no $\hb$ path created from $T_{i}^{\neg A(x_i)}$ to the thread $T_f$. Therefore, the read event $\rd(s_i, 2)$ can read from the write in $T_i^{A(x_i)}$ and $\rd(s_i, 1)$ can read from the unused thread $T_i^{\neg A(x_i)}$. This gives an $\rf$ which satisfies weak-read-coherence. By design of the execution graph, the $\rf$ also satisfies porf-acyclicity. However, for the lemma, we need to show consistency for $\sramm$, hence we require to prove strong-write-coherence and read-coherence. For this we need to also synthesize an $\mo$. 
Consider any  $\mo$ relation that respects the following orders:
    \begin{itemize}[leftmargin=0.1em]
    \item $T_{i}^{A(x_i)}: \wt(s_i, 1) \xra{\mo} T_i^{\neg A(x_i)}: \wt(s_i, 1) \xra{\mo}T_{i}^{\neg A(x_i)}: \wt(s_i, 2) \xra{\mo} T_i^{A(x_i)}: \wt(s_i, 2)$,
    \item For each clause $j$ there are three threads writing $\wt(c_j, 1)$.  For each clause $C_j$ choose a specific literal $\ell_j$ that is made true by the satisfying assignment $A$. The $\mo$ relation orders the other write statements $\wt(c_j, 1)$ before $T_{\ell_j}: \wt(c_j, 1)$. 
    \end{itemize}  
It can be shown that the constructed $\rf$ and $\mo$ satisfy strong-write-coherence and read-coherence.

\paragraph*{Strong-write-coherence.} We will show that the constructed $\rf$ and $\mo$ satisfy strong-write-coherence, that is acyclicity of $\hb \cup \mo$. Suppose not. There is a cycle $\mathcal{C}$ consisting of $\po$, $\rf$ and $\mo$ edges (thereby violating strong-write-coherence). Since we have seen that there is no $\po \cup \rf$ cycle, cycle $\mathcal{C}$ should contain an $\mo$ edge.  An edge  $\tuple{T_{i}^1: \wt(s_i, 1)} ~\mo~\tuple{T_i^1: \wt(s_i, 1)}$ cannot be present in $\mathcal{C}$ since there can be no incoming $\po$ or $\rf$ edge to $\tuple{T_{i}^1: \wt(s_i, 1)}$. Now, consider an edge $\tuple{T_{i}^0: \wt(s_i, 2)}~\mo~\tuple{T_{i}^1: \wt(s_i, 2)}$. The only incoming edge to $\tuple{T_{i}^0: \wt(s_i, 2)}$ is a $\po$ edge from $\tuple{T_{i}^0: \wt(s_i, 1)}$, and there can be no $\po$ or $\rf$ edges into the latter event. Hence, even the second type of $\mo$ edges cannot be present in $\mathcal{C}$. 
    This leaves us with $\mo$ edges of the form $\tuple{T_{\ell'}: \wt(c_j, 1)} ~\mo~ \tuple{T_{\ell_j}: \wt(c_j,1)}$. Notice that from the latter event, $\tuple{T_{\ell_j}: \wt(c_j, 1)}$ the only outgoing edges are to $T_f$, and from $T_f$, there are no outgoing edges, hence there is no path from an event in $T_f$ to $\tuple{T_{\ell'}: \wt(c_j, 1)}$. This proves that such an $\mo$ edge cannot be part of an $\hb \cup \po$ cycle too. We conclude that there are no $\hb \cup \mo$ cycles.

    \paragraph*{Read-coherence.} Finally, we need to prove read-coherence, i.e., $\irr(\rf^{-1}; \mo_x; \hb)$. If read coherence is violated, there is a read event $r$, and two write events $w, w'$ such that $w~\rf~r$, $w~\mo~w'$ and $w'~\hb~r$. Such a violation can potentially happen on variables which have multiple write events: in our case $s_i$ and $c_j$.
    To do this, let us collect the view of each read event on these variables and note that no read observes a value that is $\mo$-overwritten in its happens-before past. 
    \begin{itemize}[leftmargin=0.1em]
    \item By construction $\tuple{T_f: \rd(s_i, 2)}$ contains $T_i^{A(x_i)}$ and does not contain $T_i^{\neg A(x_i)}$. Therefore there is no write event $\mo$-later than $T_i^{A(x_i): \wt(s_i, 2)}$ in this view. 
    \item For the read event $\tuple{T_f: \rd(s_i, 1)}$, all writes to $s_i$ except $\tuple{T_i^{\neg{A(x_i)}}: \wt(s_i, 2)}$ are visible to it. Moreover, the event $\tuple{T_i^{\neg A(x_i)}: \wt(s_i, 1)}$ that writes to it is the $\mo$-latest in the view.
    \item For the variable $c_j$, the view of $\tuple{T_f: \rd(c_j, 1)}$ contains only the write $\wt(c_j, 1)$ from the thread $T_{\ell_j}$ corresponding to the literal $\ell_j$ that satisfies $C_j$ under $A$. All other writes to $c_j$ occur before this write in $\mo$.
    \end{itemize}
    
    This shows that the constructed $\rf$ and $\mo$ are consistent with $\sramm$.
\end{proof}

\subsection{Missing proofs from Section~\ref{sec:hardness-twowriter}}
\label{sec:app-np-hardness}

\begin{lemma}
If  $\ypartial_\varphi \models \wramm$, then $\varphi$ has a satisfying assignment. 
\end{lemma}
\begin{proof}
The proof proceeds similar to Lemma~\ref{lem:wra-eth}. We need to argue that if there is an $\rf$ that satisfies weak-read-coherence, then there cannot be two distinct reads $\rd(d_{j}, 1)$ and $\rd(d_{p}, 1)$ in $T_f$ that (directly or indirectly) read from threads corresponding to both a literal and its negation for some variable $x_i$. The $\rd(d_j, 1)$ event in $T_f$ can either read from  $T_{\ell_j^3}$ or from $T_{j}$. Similarly, $\rd(d_p, 1)$ can read from $T_{\ell_p^3}$ or $T_p$. If it so happens that the two threads from which $\rd(d_{j}, 1)$ and $\rd(d_{p}, 1)$ read from are of the form $T_{x_i}$ and $T_{\neg x_i}$, then both $\tuple{T_i^0: \wt(s_i, 2)}$ and $\tuple{T_i^1: \wt(s_i, 2)}$ are in the view of $\tuple{T_f: \rd(s_i, 2)}$. Then, since the execution satisfies weak-read-coherence, there is no write event available for $\tuple{T_f: \rd(s_i, 1)}$. A contradiction. Hence the set of literals making the clauses true in $T_f$ encode a satisfying assignment. 
\end{proof}

\begin{lemma}
If $\varphi$ has a satisfying assignment, then $\ypartial_\varphi^{\ramm} \models \sramm$. 
\end{lemma}
\begin{proof}
Let $A$ be a satisfying assignment. The $\rf$ and $\mo$ relations on events with variables $s_i$ remains the same as in the proof of Lemma~\ref{lem:sat-implies-sra}. There are changes for $c_j$ and $d_j$. Consider a clause $C_j$.
\begin{itemize}
\item If $A$ makes $\ell_j^3$ true then:
\begin{itemize}
\item $\tuple{T_{\ell_j^3}: \wt(d_j, 1)}~\rf~\tuple{T_f: \rd(d_j, 1)}$,
\item Choose an arbitrary $\wt(c_j, 1)$ for $\tuple{T_j: \rd(c_j, 1)}$
\end{itemize}
\item Else, $A$ makes one of the first two literals of $C_j$ true. Arbitrarily pick one $\ell_j^i$, for $i \in \{1,2\}$:  
\begin{itemize}
\item $\tuple{T_{\ell_j^i}: \wt(c_j, 1)}~\rf~ \tuple{T_j: \rd(c_j, 1)}$,
\item $\tuple{T_{j}:\wt(d_j, 1)}~\rf~\tuple{T_f: \rd(d_j, 1)}$
\end{itemize}
\end{itemize}
For every variable-value pair, there are atmost two writes. With the $\rf$ picked as above, the $\mo$ relation gets fixed: the unused write is $\mo$-before the used write.

Clearly, the $\rf$ relation agrees with $\po$ and hence does not violate porf-acyclicity. We need to prove strong-write-coherence and read-coherence as in the proof of Lemma~\ref{lem:sat-implies-sra}.

\paragraph*{Strong-write-coherence.} 
Suppose, for contradiction, that there is a cycle in $\hb \cup \mo$. As in Lemma~\ref{lem:sat-implies-sra}, the only possible $\mo$ edges are between the two writes to $s_i$, $c_j$, or $d_j$. For $s_i$, the argument is identical to before: the only incoming edges to the earlier write are from $\po$, and there are no cycles. For $c_j$ and $d_j$, each has at most two writers, and the $\mo$-later event leads to $T_f$ from which there is no path back to the $\mo$-earlier event. Thus, no cycle can be formed. Therefore, strong-write-coherence holds.

\paragraph*{Read-coherence.}
Suppose, for contradiction, that there is a violation of read-coherence: there exists a read $r$ and two writes $w, w'$ to the same variable such that $w~\rf~r$, $w~\mo~w'$, and $w'~\hb~r$. As before, this can only occur for variables $s_i$, $c_j$, or $d_j$. For $s_i$, the argument is unchanged from Lemma~\ref{lem:sat-implies-sra}. For $c_j$ and $d_j$, the view of each read contains only the $\rf$-used write, and the other write (if any) is $\mo$-before and not visible to the read. Thus, no such violation can occur.

Therefore, the constructed $\rf$ and $\mo$ are consistent with $\sramm$. 
\end{proof}

\section{$\rlxmm$ Memory Model}
\label{app:relaxed}
In this appendix we study the complexity of the $\vm$- problem for the `relaxed' fragment of C-11. Relaxed model satisfy the coherence axioms: \RlxRCoh{} and \RlxWCoh{} (Table~\ref{table:ax:models-2-app}). Essentially all $\hb$'s in \WRCoh{} and \WCoh{} are replaced by $\po$. For better unerstanding, Figure~\ref{fig:relaxed-illust-app} shows instances for violation of the relaxed axioms. For the stronger variation of the relaxed model, $\rlxmm$-Acyclic we have the extra requirement of  \PORF{}.  The overview of this section is as follows:
\begin{itemize}
 \item Appendix~\ref{sec-onewriter-relaxed} gives a  polynomial time algortihm for the $\vm$-problem restricted to  \onewriter~ under $\rlxmm$-Acyclic memory model.
 
 \item Appendix~\ref{sec:finegrained-relaxed} gives a conditional lower bound for the $\vm$-problem  under ETH for \threewriter.
 
 \item Appendix~\ref{sec:np-relaxed} shows that the $\vm$-problem restricted to \twowriter~ is NP-hard
\end{itemize}

\begin{table}
    \caption{\small Coherence Axioms for $\rlxmm$-Acyclic variant of C11}
    \label{table:ax:models-2-app}
    \centering
    \begin{tabular}{lr}
    \multicolumn{2}{c}{\underline{Relaxed-Acyclic}} \\
    $\irr(\po\cup\rf)$ & (\PORF{})\\
    $\irr(\mo_x;\po)$ & (\RlxWCoh{})\\
    $\irr(\rf^{-1};\mo_x;\rf^?;\po)$ & (\RlxRCoh{})\\
    \end{tabular}
\end{table}

 \begin{figure}[t]
    \def\ystep{0.4}
    \centering
     \resizebox{0.5\textwidth}{!}{
	\begin{tikzpicture}[yscale=1]
	
      \node (t11) at (12.3,0*\ystep)  {$\wt_1(x)$};
      \node (t12) at (12.3,-4*\ystep) {$\wt_2(x)$};
      \node (t0) at (12.3,-5*\ystep) {$(a)$};
      \draw[mo] (t12) to[bend left=45] node[left]{$\mo$} (t11);
      \draw[po] (t11) to (t12);
     \node (t11) at (16.3,0*\ystep)  {$\wt_1(x)$};
      \node (t12) at (16.3,-4*\ystep) {$\wt_2(x)$};
      \node (t21) at (18.1,0*\ystep) {$\rd_2(x)$};
      \node (t22) at (18.1,-4*\ystep) {$\rd_1(x)$};
      \node (t1) at (17.2,-5*\ystep) {$(b)$};
      \draw[mo] (t11) to node[left]{$\mo$} (t12);
      \draw[po] (t21) to (t22);
      \draw[rf,bend left=0] (t12) to node[above, pos=.8, sloped]{$\rf$}(t21);
      \draw[rf,bend right=0] (t11) to node[below,pos=0.8, sloped]{$\rf$} (t22);
    
    \end{tikzpicture}
	 }
    \caption{Violation of (a)~\RlxWCoh:$\wt_2(x)~\mo~\wt_1(x), \wt_1(x)~\po~\wt_2(x)$ (b)~\RlxRCoh{}: $\wt_1(x)~\mo~\wt_2(x)$, $\rd_2(x)~\po~\rd_1(x)$ 
      with $\wt_2(x)~\rf~\rd_2(x)$ and $\wt_1(x)~\rf~\rd_1(x)$.}
      \label{fig:relaxed-illust-app}
\end{figure}

\subsection{Polynomial Time for \onewriter\ systems}\label{sec-onewriter-relaxed}
Here we give a polynomial time algorithm for solving  the $\vm$-problem for \onewriter~ under $\rlxmm$-Acyclic memory model. Note that as $\rlxmm$-Acyclic is stronger than $\rlxmm$. So any polynomial time result for $\rlxmm$-Acyclic would also hold for $\rlxmm$.
As for \onewriter~ systems we have $\mo=\po$, \RlxWCoh{} would always hold. Hence to show that an execution graph $\expartial$ is $\rlxmm$-Acyclic-consistent, we need to verify \PORF{} and \RlxRCoh{}. We can state the following lemma for $\rlxmm$ similar to  Lemma~\ref{lem:onewriter-axioms-two-properties}.

\begin{restatable}{lemma}{lemmaonewriterrelaxedtwoproperties}
    \label{lem:relaxed-onewriter-axioms-two-properties}
Let $\rf_1$ and $\rf_2$ be two reads-from functions over a partial \onewriter\ execution graph $\ex=(\E, \po, \mo)$ s.t. $\mo$ agrees with $\po$. 
\begin{itemize}
\item Suppose $\rf_1 \sqsubseteq \rf_2$. Then, if $\rf_1$ induces a $\po \cup \rf_1$ cycle, then $\rf_2$ induces a $\po \cup \rf_2$ cycle.
\item Suppose both $\rf_1$ and $\rf_2$ satisfy relaxed-read-coherence. Then, the $\rf$ function $\min(\rf_1, \rf_2)$ also satisfies relaxed-read-coherence.
\end{itemize} 
\end{restatable}
\begin{proof}
 The proof for the first part is same as that of $\wramm$ in lemma~\ref{lem:onewriter-axioms-two-properties} and hence we skip it.
 
 For the second part, we denote $\rf_0=min(\rf_1,\rf_2)$. We shall show that that if $\rf_0$ violates relaxed-read-coherence, then either $\rf_1$ or $\rf_2$ violates it. Suppose $\rf_0$ induces a violation of relaxed-read-coherence. There is a cycle of the form:
 \[ r \xra{\rf_0^{-1}} w \xra{\po^+} w' \xra{\rf_0} r' \xra{\po^+} r\], where all of $w,w',r$ and $r'$ are in the same variable.
 By definition, either $\rf_1$ or $\rf_2$ associates $r$ to $w$. WLOG, say $\rf_1(r) = w$. Then either $w'~\rf_1~r'$ or $w'~\rf_2~r'$. For the former case, $\rf_1$ is not $\rlxmm$-consistent. For the later case,  there is a write event $w^{''}$ on $\xvar$ such that $w'~\po^*~w{''}~\rf_1~r'$ as $\rf_0$ maps $r'$ to the $\po$-earlier read between $\rf_1(r')$ and $\rf_2(r')$. Again relaxed-read-coherence of $\rf_1$ fails thereby proving the second item. 
\end{proof}

\textbf{Finding the smallest $\rlxmm$-consistent $\rf$}
To show that an execution graph, $\expartial$ satisfies $\rlxmm$ consistency, we check for \PORF{}, \RlxRCoh{} and \RlxWCoh{}. 

Algorithm~\ref{algo:one-writer-wra} can be modified to obtain a PTIME algorithm for $\rlxmm$. The modifications are as follows:
\begin{itemize}
 \item Any appearance of \WRCoh{} would be replaced  by \RlxRCoh{}.
 
 \item Any appearance of $\wramm$-consistency would be replaced by $\rlxmm$-consistency.
 
 \item Line~20 would be replaced by `Pick a cycle violating relaxed-read-coherence: $r_i \xra{\rf_i^{-1}} w_i \xra{po} w'_i \xra{\rf} r'_i \xra{po} r_i$'. 
 
 \item We use check for \PORF{} in line~\ref{algo-check-porf} iof Algorithm~\ref{algo:one-writer-wra} only if the underlying memory model is $\rlxmm$-Acyclic.
 
 \end{itemize}

Note that, by construction of the algorithm, \RlxWCoh{} is never violated. The correctness of the algorithm uses Lemma~\ref{lem:relaxed-onewriter-axioms-two-properties} and is analogous to that of Algorithm~\ref{algo:one-writer-wra}. Theorem~\ref{th-wra-correctness} can be restated in terms of $\rlxmm$ as follows.

\begin{restatable}{theorem}{thmrlxcorrectness}\label{th-wra-correctness}
A partial execution graph $\expartial = (\E, \po)$ of a one-writer system is consistent for $\rlxmm$ iff the modified version of Algorithm~\ref{algo:one-writer-wra} for $\rlxmm$ memory model, executes Line~$6$. Moreover, the algorithm runs in time $\mathcal{O}(n^3)$ where $n = |E|$  and returns the smallest $\rf$ function that is $\rlxmm$-consistent, if one exists.
\end{restatable}

The calculation of complexity of the algorithm is similar to that of $\wramm$ (Section~\ref{sec:onewriter-algorithm}). We want to check for a cycle of the form  $r \xra{\rf_0^{-1}} w \xra{\po^+} w' \xra{\rf_0} r' \xra{\po^+} r$. 
For each read event $r$ and its corresponding write event $w$, we need to check if there exist a $w'$ such that (1) there is a path of $\po$ edges from $w$ to $w'$,(2) $w'~\rf~r'$ and (3) $r'~\po^+~r$. For each read, this can be achieved in $\mathcal{O}(n)$ time by a graph search. Hence, overall, across all possible reads, this algorithm takes $\mathcal{O}(n^2)$ time. Therefore, each call to update costs $\mathcal{O}(n^2)$ and there can be at most $n$ calls. This gives a complexity of $\mathcal{O}(n^3)$.

\subsection{Fine-grained hardness under $ETH$ for \threewriter\ systems}
\label{sec:finegrained-relaxed}

\begin{lemma}\label{lem:relaxed-finegrained}
    Under ETH, there is no $2^{o(k)} \cdot n^{O(1)}$ algorithm for the consistency testing  problem for $\rlxmm$ where $k$ is the number of threads, and $n$ the number of events, even when the executions graphs are \threewriter~ and the number of writes per variable is four. 
\end{lemma}

Let $\varphi=\{C_i\}_{i \in [m]}$ be a Boolean formula over $k$ variables 
$\{x_j\}_{j \in [k]}$ and $m$ clauses of the form $C_j=(\ell_j^1, \ell_j^2, \ell_j^3)$ with where each literal is either a variable $x_i$ or its negation $\neg x_i$. 
Given a 3-CNF-SAT formula $\varphi$ with $k$ variables and $m$ clauses, we construct an abstract execution $\expartial=(\E, \po)$ with $2k+3$ threads, accessing $k+m$ memory locations $v_1, v_2, \dots, v_k, c_1, \dots, c_m$ and 2 values $1, 0$ such that $\varphi$ is satisfiable iff $\expartial_\varphi \models \rlxmm$. 
$\expartial_\varphi$ follows the general scheme depicted in Table \ref{tab:reduction-threads-relaxed}.

We remark that an $\mo$ edge from $\wt(v_i, 0)$ to $\wt(v_i,1)$ encode variable $x_i$ being set to \emph{true}, and an $\mo$ edge in the opposite direction $x_i$ being set to \emph{false}.

\begin{wrapfigure}{r}{0.5\textwidth}
  \vspace{-\baselineskip}
  \centering
  \includegraphics[width=0.3\textwidth]{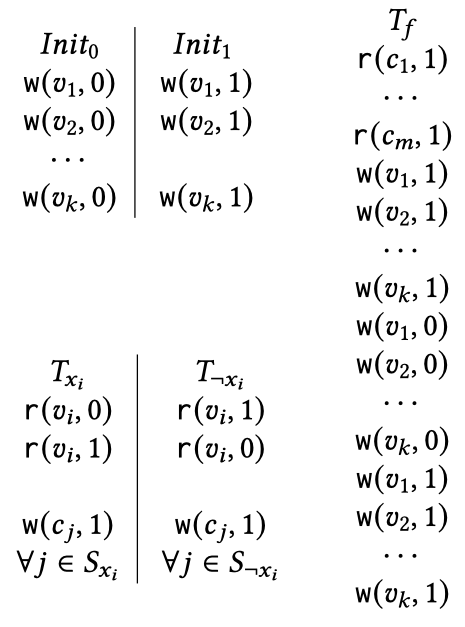}
    \caption{Threads in the reduction for $\rlxmm$}
    \label{tab:reduction-threads-relaxed}
\vspace{-\baselineskip}
\end{wrapfigure}
The threads $\mathsf{Init}{1}$ and $\mathsf{Init}{0}$ initialize the valuation used to evaluate the formula. For each variable $x_i$ corresponding to the Boolean variables of $\varphi$, $\mathsf{Init}{1}$ writes $1$ and $\mathsf{Init}{0}$ writes $0$. 
$T_{x_i}$ has a read statement $\rd(v_i,0)$ followed by $\rd(v_i,1)$ (which denotes reading \emph{true} for $x_i$) and sets all clauses containing $x_i$ to true, using the statement $\wt(c_j, 1)$ for all $j \in S_{x_i}$.

Similarly, thread $T_{\neg x_i}$ reads $\rd(v_i, 1)$ followed by $\rd(v_i, 0)$ (denoting the value \emph{false} for $x_i$) and sets all clauses containing $\neg x_i$ to true. 
Finally the thread $T_{f}$ needs to make sure that all clauses are set to true, accomplished by the read statements $\rd(c_1, 1), \dots, \rd(c_m, 1)$.
Then, $T_{f}$ proceeds through the following steps (1) all the variables $v_i$s are set to $0$, followed by (2) setting all the variables $v_i$s to $1$, and finally (3) setting all the variables $v_i$s to $0$ again.
This sequence ensures that any threads $T_{x_i}$, $T_{\neg x_i}$ which were not executed in Phase 2 now have their conditions satisfied and can complete their execution.

\begin{lemma}
    $\varphi$ is satisfiable iff $\expartial_\varphi  \models \rlxmm$.
\end{lemma}

\begin{proof}
    Suppose $A$ is a satisfying assignment for $\varphi$. We explictly describe the $\rf$ relation and $\mo$ relation as follows. 
    
    If $A(x_i) = true$:
    \begin{itemize}[leftmargin=0.1em]
    \item $\tuple{Init_0: \wt(v_i, 0)}~\rf~\tuple{T_{x_i}: \rd(v_i, 0)}$ and $\tuple{Init_1: \wt(v_i, 1)}~\rf~\tuple{T_{x_i}: \rd(v_i, 1)}$,
    \item $\tuple{T_f: \wt(v_i, 1)}~\rf~\tuple{T_{\neg x_i}: \rd(v_i, 1)}$, and $\tuple{T_f: \wt(v_i, 0)}~\rf~\tuple{T_{\neg x_i}: \rd(v_i, 0)}$. 
    \end{itemize} 
    Else, when $A(x_i) = false$, 
    \begin{itemize}[leftmargin=0.1em]
    \item $\tuple{Init_1: \wt(v_i, 1)}~\rf~\tuple{T_{\neg x_i}: \rd(v_i, 1)}$, and $\tuple{Init_0: \wt(v_i, 0)}~\rf~\tuple{T_{\neg x_i}: \rd(v_i, 0)}$,
    \item $\tuple{T_f: \wt(v_i, 0)}~\rf~\tuple{T_{ x_i}: \rd(v_i, 0)}$, and $\tuple{T_f: \wt(v_i, 1)}~\rf~\tuple{T_{ x_i}: \rd(v_i, 1)}$. 
    \end{itemize} 
            
    Coming to variables $c_i$s, for each $T_f: \rd(c_j, 1)$ there is a literal $\ell_j$ in clause $C_j$ that is true. Let $\tuple{T_{\ell_j}: \wt(c_j, 1)}~\rf~ \tuple{T_f: \rd(c_j, 1)}$. Observe that the constructed $\rf$ satisfies $\po-\rf$-acyclicity.
    
    The $\rf$ induces the following $\mo$ relation. If $A(x_i) = true$: 
    \begin{itemize}[leftmargin=0.1em]
    \item $\tuple{Init_0: \wt(v_i, 0)} ~\mo~\tuple{Init_1: \wt(v_i, 1)} ~\mo~\tuple{T_f: \wt(v_i, 1)}~\mo~\tuple{T_f: \wt(v_i, 0)} ~\mo~\\ \tuple{T_f: \wt(v_i, 1)}$
    \item For each clause $C_j$ there are three threads writing $\wt(c_j, 1)$.  For $C_j$, choose a specific literal $\ell_j$ that is made true by the satisfying assignment $A$. 
    The $\mo$ relation orders the write statements $\wt(c_j, 1)$ in literals that are set to true before $T_{\ell_j}: \wt(c_j, 1)$, and those in literals that set to false after $T_{\ell_j}: \wt(c_j, 1)$.
    \end{itemize}  
    When $A(x_i) = false$, an $\mo$ order can be defined analogously. 

    \paragraph*{Relaxed Write-coherence.} From our construction, it is easy to see $\irr(\mo_x; \po)$.

    \paragraph*{Relaxed Read-coherence.} We need to prove  $\irr(\rf^{-1}; \mo_x; \rf^{?}; \po)$. 
    If relaxed read coherence is violated, the violating pattern is of the following form there are read events $r_1$ and $r_2$, and two write events $w_1, w_2$ such that $w_1~\rf~r_1$, $w_2~\rf~r_2$, $w_1~\mo~w_2$ and $r_2~\po~r_1$.
    Such a violation can potentially happen on variables which have multiple read events in the same thread, which is only the case for $v_i$.
    We consider the various possibilities below. 
    We consider the case $A(x_i) = true$ below. 
        \begin{itemize}[leftmargin=0.1em]
            \item $\tuple{T_{x_i}: \rd(v_i, 0)} \po \tuple{T_{x_i}: \rd(v_i, 1)}$, reads-from $Init_0$ and $Init_1$ respectively, where $\mo$ agrees with the $\po$.
            \item $\tuple{T_{\neg x_i}: \rd(v_i, 1)} \po \tuple{T_{\neg x_i}: \rd(v_i, 0)}$, reads-from the first two writes on $v_i$ in $T_f$ respectively, where, once again $\mo$ agrees with the $\po$, and irreflexivity is ensured.
        \end{itemize}
        A similar argument can be made for the case the case $A(x_i) = false$.

    This shows that the constructed $\rf$ and $\mo$ are consistent with $\rlxmm$. 
    
    Conversely, suppose $\expartial_\varphi \models \rlxmm$. 
    We need to show that there are no two $\rd(c_j, 1)$ and $\rd(c_\ell, 1)$ reading from some $T_{x_i}$ and $T_{\neg x_i}$.
    Suppose not. 
    Let $\rd(c_j, 1)$ read from $T_{x_i}: \wt(c_j, 1)$ and $\rd(c_\ell, 1)$ read from $T_{\neg x_i}: \wt(c_\ell, 1)$.    
    Consider the thread $T_{x_i}$. Since $T_{x_i}: \wt(c_j, 1)$ has been read from, this means that the preceding read instructions of $T_{x_i}$ have been executed, and the only possible read from relation is as follows: 
    $\tuple{Init_0: \wt(v_i, 0)}~\rf~\tuple{T_{x_i}: \rd(v_i, 0)}$,
    and $\tuple{Init_1: \wt(v_i, 1)}~\rf~\tuple{T_{x_i}: \rd(v_i, 1)}$. 
    By relaxed read coherence axiom, this implies $\tuple{Init_0: \wt(v_i, 0)} ~\mo~\tuple{Init_1: \wt(v_i, 1)}$    

    Now, consider the thread $T_{\neg x_i}$. As argued fro $T_{x_i}$, since the last instruction has been read from, the preceding read instructions of $T_{x_i}$ have been executed with the only reads-from relation being 
    $\tuple{Init_1: \wt(v_i, 1)}~\rf~\tuple{T_{\neg x_i}: \rd(v_i, 1)}$ and $\tuple{Init_0: \wt(v_i, 0)}~\rf~\tuple{T_{\neg x_i}: \rd(v_i, 0)}$. 
    However, this violates relaxed read coherence axiom, as $\tuple{Init_0: \wt(v_i, 0)} ~\mo~\tuple{Init_1: \wt(v_i, 1)}$ but $\tuple{T_{\neg x_i}: \rd(v_i, 1)}~\po~\\ \tuple{T_{\neg x_i}: \rd(v_i, 0)}$. 
    This contradicts our assumption that there are two $\rd(c_j, 1)$ and $\rd(c_\ell, 1)$ reading from some $T_{x_i}$ and $T_{\neg x_i}$. 
    Thus, the assignment induced by the $T_i^b$ threads executed before $T_f$ satisfies $\varphi$.
\end{proof}

\subsection{NP-hardness for \twowriter\ systems}
\label{sec:np-relaxed}

Given a 3-CNF-SAT formula $\varphi$ with $k$ variables and $m$ clauses, we construct an abstract execution $\expartial=(\E, \po)$ with $2k+m+3$ threads, accessing $k+2m+1$ memory locations $v_1, v_2, \dots, v_k, c_1, \dots, c_m, d_1, \dots, d_m, s$ and 2 values $0, 1$ such that $\varphi$ is satisfiable iff $\ypartial_\varphi^{\rlxmm}  \models \rlxmm$.
$\ypartial_\varphi^{\rlxmm}$ follows the general scheme depicted in Table \ref{tab:np-reduction-threads-relaxed}.

\begin{table}[h]
    
\noindent
\begin{minipage}{0.3\textwidth}
\centering
\begin{tabular}{c | c}
    $Init_0$ & $Init_1$ \\
    $\wt(v_1, 0)$~& ~$\wt(v_1, 1)$ \\
    $\wt(v_2, 0)$~ & ~$\wt(v_2, 1)$ \\
    $\cdots$ \\
    $\wt(v_k, 0)$~ & ~$\wt(v_k, 1)$\\
    & $\rd(s, 1)$ \\
    & $\wt(v_1, 0)$ \\
    & $\wt(v_2, 0)$ \\
    & $\cdots$ \\
    & $\wt(v_k, 0)$ \\
    & $\wt(v_1, 1)$ \\
    & $\wt(v_2, 1)$ \\
    & $\cdots$ \\
    & $\wt(v_k, 1)$ \\
\end{tabular}
\end{minipage}
\hfill
\begin{minipage}{0.3\textwidth}
\centering
\begin{tabular}{c | c}
    $T_{x_i}$~ & ~$T_{\neg x_i}$ \\
        $\rd(v_i, 0)$ ~&~  $\rd(v_i, 1)$ \\
        $\rd(v_i, 1)$ ~&~ $\rd(v_i, 0)$ \\
        & \\        
        $\wt(c_j/d_j, 1)$ ~&~ $\wt(c_j/d_j, 1)$ \\
        $\forall j \in S_{x_i}$ ~&~  $\forall j \in S_{\neg x_i}$\\
        $d_j$ if $x_i$ is ~&~   the third literal\\
        $c_j$  ~&~  otherwise
\end{tabular}
\end{minipage}
\hfill
\begin{minipage}{0.15\textwidth}
\centering
\begin{tabular}{c}
    $T_j$ \\
    $\rd(c_j, 1)$\\
    $\wt(d_j, 1)$ \\
    $\forall j \in \{1,\cdots,m\}$\\
\end{tabular}
\end{minipage}
\hfill
\begin{minipage}{0.15\textwidth}
\centering
\begin{tabular}{c}
    $T_f$ \\
    $\rd(d_1, 1)$ \\
    $\rd(d_2, 1)$ \\
    $\cdots$ \\
    $\rd(d_m, 1)$ \\
    $\wt(s, 1)$ \\
\end{tabular}
\end{minipage}
\vspace{2em}
\caption{Threads in the reduction for $\rlxmm$}
    \label{tab:np-reduction-threads-relaxed}
 \vspace{-.8cm}   
\end{table}

\begin{lemma}
    $\varphi$ is satisfiable iff $\ypartial_\varphi^{\rlxmm}  \models \rlxmm$.
\end{lemma}

\begin{proof}
    Suppose $A$ is a satisfying assignment for $\varphi$. Then, the execution proceeds as follows. 
    The $\rf$ relation and $\mo$ relation are exactly as in Lemma~\ref{lem:relaxed-finegrained}, except the following.
    If $A(x_i) = true$:
    \begin{itemize}
        \item $\tuple{T_f: \wt(s, 1)}~\rf~\tuple{Init_1: \rd(s, 1)}$,
        \item $\tuple{Init_1: \wt(v_i, 1)}~\rf~\tuple{T_{\neg x_i}: \rd(v_i, 1)}$, and $\tuple{Init_1: \wt(v_i, 0)}~\rf~\tuple{T_{\neg x_i}: \rd(v_i, 0)}$. 
    \end{itemize} 
    For clause variables $c_j$ and $d_j$ , the $\rf$ (and consequently $\mo$) relation depends on the literals set to true by $A$. 
    For each clause $C_j$ there are two threads writing $\wt(c_j, 1)$, and two threads writing $\wt(d_j, 1)$. 
    
    If $A$ makes $\ell_j^3$ true then $\tuple{T_{\ell_j^3}: \wt(d_j, 1)}~\rf~\tuple{T_f: \rd(d_j, 1)}$. Otherwise $A$ makes $\ell_j^1$ or $\ell_j^2$ true. Let $\ell_j^k$ be the literal that is made true by the satisfying assignment $A$. Then, we have $\tuple{T_{\ell_j^k}: \wt(c_j, 1)}~\rf~\tuple{T_j: \rd(c_j, 1)}$ and $\tuple{T_j: \wt(d_j, 1)}~\rf~\tuple{T_f: \rd(d_j, 1)}$.

    For every variable-value pair for $c_j$ and $d_j$, there are atmost two writes. With the $\rf$ picked as above, the $\mo$ relation gets fixed: the unused write in $T_{x}$/ $T_{\neg x}$ that is set to true by $A$ is $\mo$-before the used write, while if the unused write is in $T_{x}$/ $T_{\neg x}$ that is set to false by $A$, then it is $\mo$-after the used write.    
    
    \paragraph*{Relaxed Write-coherence.} From our construction, it is easy to see $\irr(\mo_x; \po)$.

    \paragraph*{Relaxed Read-coherence.} We need to prove  $\irr(\rf^{-1}; \mo_x; \rf^{?}; \po)$. 
    The possible violating patterns and the arguments follow as in Lemma~\ref{lem:relaxed-finegrained}.

    This yields an $\rf$ where all $\rd(d_j, 1)$ in $T_f$ read from a thread corresponding to a true literal in $A$, so $\ypartial_\varphi^{\rlxmm} \models \rlxmm$.

    Conversely, suppose $\ypartial_\varphi^{\rlxmm} \models \rlxmm$. 
    The proof proceeds as in Lemma~\ref{lem:relaxed-finegrained}, where we show that if there are two clauses $C_j$ and $C_{\ell}$ such that $\rd(d_j, 1)$ and $\rd(c_\ell, 1)$ reading from some $T_{x_i}$ and $T_{\neg x_i}$, then relaxed read coherence is violated. .
\end{proof}

\section{Causal Consistency Models}
\label{app:cc-models}

In this appendix, we study the consistency checking problem under the \emph{causal consistency models} namely, $\ccmm, \cmmm$ and $\cvmm$ which come from distributed databases.
We begin with a brief overview of the
operational and axiomatic semantics of these models, and then show that the complexity landscape observed for the Release - Acquire (RA) variants extends to these models as well. 
Recall that for RA, we established a tight characterization of complexity with respect to the number of writer threads per memory location:
\onewriter\ admits polynomial-time
solutions, \twowriter\ is $\NP$-hard, and \threewriter\ yields a stronger
conditional lower bound. Our goal in this section is to establish analogous
results for causal consistency models, thereby showing that the same complexity results extend to $\ccmm, \cmmm$ and $\cvmm$. Next, we provide an overview of these causal consistency models.

\subsection{Overview of Causal Consistency Models}
\label{app:cc-models-overview}

In the domain of distributed systems, causality is the key concept for consistency. There are three popular causal models, namely, Causal Consistency ($\ccmm$), Causal Convergence ($\cvmm$), and Causal Memory ($\cmmm$)~\cite{BouajjaniEGH17}. 
It is known that $\ccmm$ coincides with $\wramm$ while $\cvmm$ coincides with $\sramm$~\cite{Lahav:2022}. $\cmmm$ is stronger than $\ccmm$ and incomparable to $\cvmm$ and $\ramm$.  

In Causal Memory, each thread has a locally-consistent view of the order that different writes have been executed~\cite{Ahamad1995,BouajjaniEGH17}. 
This is formalised by introducing an additional relation for each thread, called the ``happened-before'' relation in the literature.  To avoid confusion with the ``happens before'' $\hb$, we call this the \emph{observed relation} $\ob{}$. 

To enhance understanding we define the notion of \emph{conflicting triplets} first.

\begin{definition}(Conflicting Triplet)
Given an execution graph $\expartial=(\E,\po)$, two events in $\E$ are called \emph{conflicting} if both are on the same variable and at least one of them is a write. If also the reads-from function $\rf$ is given, a \emph{Conflicting Triplet} is defined as a tuple $(\wt, \rd, \wt')$ of pairwise conflicting events in $\E$ such that $\wt~\rf~\rd$. 
\end{definition} 
 
\noindent{The observed relation $\ob{}$.}
Given an event $\event$, the \emph{observed relation} for $\event$ is the smallest transitive relation $\ob{\event}\subseteq \E\times \E$ with the following properties. 

\begin{enumerate}
\item For every $(\event_1, \event_2)\in \hb$ such that $(\event_i, \event)\in \hb$ for each $i\in[2]$, we have $(\event_1, \event_2)\in \ob{\event}$.
\item For every conflicting triplets $(\wt, \rd, \wt')$ such that:
(ii)~$(\wt',\rd)\in \ob{\event}$ and
(iii)~$(\rd, \event)\in \po^?$,
we have $(\wt', \wt)\in \ob{\event}$.
\end{enumerate}

The idea is that any thread which executes the read $\rd$ must observe $\wt$ 
after it has observed $\wt'$ so that it does not read an obsolete value. 
For any event $\event$, $\ob{\event}$ contains $\hb$, and $\ob{}$ grows monotonically as we move down a thread 
in $\po$-order. That is, for any $(\event_1, \event_2)\in \po$, we have $\ob{\event_1}\subseteq \ob{\event_2}$. For a thread $t$, the observed relation  is defined as $\ob{t}=\ob{\event^{\max}}$, where $\event^{\max}$ is the $\po$-maximal event of $t$.  $\cmmm$ requires all the axioms of $\wramm$, and in addition, 
requires that $\ob{t}$ is irreflexive~\cite{BouajjaniEGH17}. Figure \ref{cm:examples} depicts 
some examples. 

\begin{itemize}
\item $(\text{PORF})\land(\text{weak-Rcoh})\land (\text{\OB{}})$ \hfill [$\cmmm$]
\end{itemize}

Thus, $\wramm\mmorder \cmmm$. However, $\cmmm$ is incomparable with $\ramm$/$\sramm$, i.e., $\cmmm$ allows executions that are inconsistent in $\ramm$/$\sramm$ and vice versa.

	\begin{figure}
	\centering
	\includegraphics[width=\textwidth]{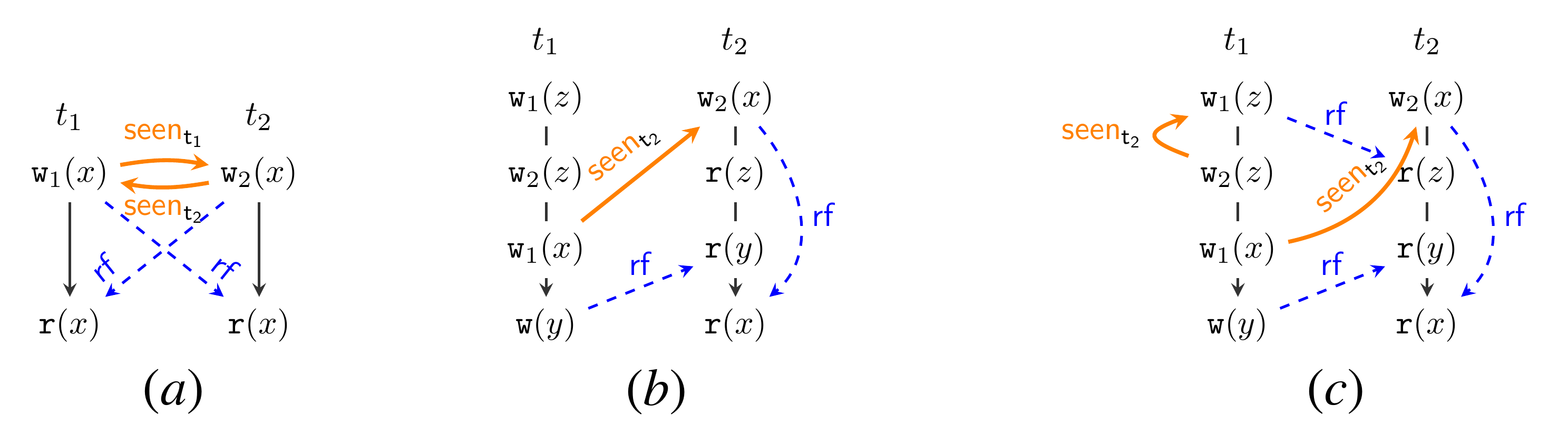}
	\caption{(a) is $\cmmm$ consistent as $\ob{t_1}, \ob{t_2}$ are acyclic. 
	(b) is also $\cmmm$ consistent : the $\wt_1(x) \ob{t_2} \wt_2(x)$ is induced 
	by $\wt_1(x)~\hb~\rd(x)$ and $\wt_2(x)~\rf~\rd(x)$. 
	(c) is obtained 
	from (b) by adding $\wt_1(z)~\rf~\rd(z)$, and is $\cmmm$ inconsistent, 
since we have $\wt_1(z)~\po~\wt_2(z)~\ob{t_2}~\wt_1(z)$. 
	$\wt_2(z)~\ob{t_2}~\wt_1(z)$ is induced 
	from the conflicting triplet $(\wt_1(z), \rd(z), \wt_2(z))$ with 
	 $\wt_2(z)~\ob{t_2}~\rd(z)$. 
	}
	\label{cm:examples}
	\end{figure}

Since $\wramm=\ccmm$ and $\sramm=\cvmm$, both the fine-grained-hardness and $\NP$-hardness result for $\wramm, \sramm$ implies the same for  $\ccmm$ and $\cvmm$. So we only need to prove the hardness results for $\cmmm$.
In Section \ref{app:cc-1w}, we discuss the polynomial-time algorithm for \onewriter\ programs and in Section~\ref{app:cc-3w}, we discuss the complexity results for \threewriter\ programs under $\cmmm$. Note that the $\NP$-hardness for \twowriter\ programs under $\cmmm$ follows from arguments analogous to those made in Section~\ref{sec:hardness-twowriter} for RA models, and therefore we do not discuss it separately.

For a $\cmmm$-consistent execution graph the following should be satisfied: $(\text{PORF})\land(\text{weak-Rcoh})\land (\text{\OB{}})$

\subsection{ \onewriter~ programs under Causal Consistency }~\label{app:cc-1w}

Recall that $\cmmm$ is essentially $\wramm$ with the extra axiom of $\ob{}$ acyclicity. We argue that the algorithm in  Section~\ref{sec:onewriter-algorithm} 
can be used for $\cmmm$ also. It is easy to see that any execution graph $\ex$  which is $\cmmm$-consistent is also $\wramm$-consistent.

The other direction follows for the \onewriter{} case. 
Suppose $\ex$ is $\wramm$-consistent but  violates $\ob{t}$ acyclicity for some thread 
 $t$. Then we have  both $\wt'~\ob{t}~ \wt$ and  $\wt~\ob{t}~ \wt'$ for writes $\wt, \wt'$ on some location $x$. 
 
 By definition of $\ob{t}$, $\wt'~\ob{t}~ \wt$ implies the existence of a conflicting triplet $(\wt, \rd, \wt')$ such that 
   $\wt~\rf ~\rd$, $\wt'~\ob{t}~\rd$, $\rd~\po^?~ e_{max}$ where $e_{max}$ is the $\po$-maximal element of $t$.
  Likewise,  $\wt~\ob{t}~ \wt'$ implies the existence of a conflicting triplet $(\wt', \rd', \wt)$ such that 
   $\wt'~\rf ~\rd'$, $\wt~\ob{t}~\rd'$, $\rd'~\po^?~ e_{max}$ where $e_{max}$ is the $\po$-maximal element of $t$. 
      Since $\wt, \wt'$  are in the same thread, either $\wt~\po~\wt'$ or 
   $\wt'~\po~\wt$. Either way, we observe one violation for \WRCoh{} wrt $\rd$ or $\rd'$. 
   Thus, inconsistency wrt $\cmmm$ results in $\wramm$ inconsistency as well. 
    Therefore, the algorithm    for $\wramm$ can be employed for $\cmmm$ as well. 
   Thus in case of \onewriter~ the time complexity of consistency testing will be same as that of $\wramm$, $O(n^3)$. 
   
   Since $\cvmm = \sramm$ and $\ccmm=\wramm$, the lower bound result under BMM hypothesis in Section~\ref{sec:onewriter} holds for both $\cvmm$ and $\ccmm$. Finally, as mentioned previously,
for the single writer case, obtaining a cyclic $\ob{t}$ for a thread $t$ results in violation of \WRCoh{}.  In the reduction given in Section~\ref{sec:onewriter}, \WRCoh{} is violated in the presence of a triangle, hence, the result also holds for $\cmmm$.

\subsection{Fine-grained Hardness under ETH for \threewriter~ systems}~\label{app:cc-3w}

We will now discuss the fine-grained hardness of \threewriter -system under $\cmmm$. 
Since  $\cmmm$ is stronger than $\wramm$, the soundness part of the reduction in 
 Lemma \ref{lem:wra-eth}  follows : Assume $\expartial_\varphi^{\ramm} \models \cmmm$. Then 
 we know $\expartial_\varphi^{\ramm} \models \wramm$. Our soundness argument shows that 
 $\varphi$ is satisfiable. 
 
 For the completeness part, assume that $\varphi$ is satisfiable. We have already synthesized an $\rf, \mo$ which respect $\sramm$(and hence $\wramm$). We have to now show that, \OB{} is preserved by the synthesized $\rf$. 

\begin{lemma}\label{lem:eth-cm}
If $\varphi$ has a satisfying assignment, then $\expartial_\varphi^{\ramm} \models \cmmm$. 
\end{lemma}
\begin{proof}
Let $A$ be the satisfying assignment for $\varphi$. 
Say $A(x_i) = true$:

(i) For all pair of events in $\E$, $e~\hb~e'$ we have $e~\ob{T_F}~e'$

(ii) For the conflicting tiplet $(\tuple{T_i^0: \wt(s_i, 1)}, \tuple{T_f: \rd(s_i, 1)}, \tuple{T_i^1: \rd(s_i, 2)})$, we have $\tuple{T_i^1: \wt(s_i, 2)}~ \ob{T_f}~ \tuple{T_f: \rd(s_i, 1)}$ and $\tuple{T_f: \rd(s_i, 1)}~\po^?~e^{T_f}_{max}$ where $e^{T_f}_{max}$ is the $\po$-maximal event in $T_f$. This implies $\tuple{T_i^1: \wt(s_i, 2)}~\ob{T_f}~\tuple{T_i^0: \wt(s_i, 1)}$. 

As there are no $\hb$-cycle and no members of $\ob{T_f}$ directed from $T_i^0$ to $T_i^1$, $\ob{T_f}$ is irreflexive.

If we consider the relation $\ob{T_{\neg x_i}}$, for the conflicting triplet $(\tuple{T_i^0: \wt(v_i, 0)},\\ \tuple{T_{x_i}: \rd(v_i,0)}, \tuple{T_i^1: \wt(v_i, 1)})$ we have $\tuple{T_i^1: \wt(v_i, 1)}~\ob{T_{\neg x_i}}~\tuple{T_i^0: \wt(v_i, 0)}$. 

Again no members of $\ob{T_{\neg x_i}}$ are directed from $T_i^0$ to $T_i^1$. As there is no $\hb$-cycle, $\ob{T_{\neg x_i}}$ is also irreflexive. 

It is easy to see that $\ob{}$ is irreflexive for all remaining threads.
Hence \OB{} is preserved by the synthesized $\rf$

\end{proof}

\subsection{NP-hardness for \twowriter~ systems}

For \twowriter~ systems the reduction from 3-SAT instance has similar construction as that of $\wramm$ and the proof of the reduction will be on similar argument as the proof of the reduction for \threewriter.

\end{document}